\documentclass[11pt]{article}%
\usepackage{mathtools}%

\usepackage{graphicx}
\usepackage{latexsym}

\newcommand{\width}{\ensuremath{\mathit{SubscriptTwoMax}}}
\newcommand{\blocks}{\ensuremath{\mathit{SubscriptOneMax}}}
\newcommand{\sigmastar}{\ensuremath{\Sigma^\ast}}
\newcommand{\wasast}{m}
\usepackage{tikzit}
\usepackage{pifont}
\usepackage{booktabs}
\usepackage{tabularx}
\usepackage{amsmath,amsfonts,amssymb,amsthm,nicefrac}
\usepackage{algorithm,algorithmic}
\usepackage{tabularx,multirow}
\usepackage{subfigure,color,epsfig}

\makeatletter%
\newcommand{\singlespacing}{\let\CS=
\@currsize\renewcommand{\baselinestretch}{1}\tiny\CS}
\newcommand{\singlespacingplus}{\let\CS=
\@currsize\renewcommand{\baselinestretch}{1.25}\tiny\CS}
\newcommand{\doublespacing}{\let\CS=
\@currsize\renewcommand{\baselinestretch}{1.75}\tiny\CS}
\newcommand{\extradoublespacing}{\let\CS=
\@currsize\renewcommand{\baselinestretch}{1.9}\tiny\CS}
\newcommand{\nicenicespacing}{\let\CS=
\@currsize\renewcommand{\baselinestretch}{1.9}\tiny\CS}
\newcommand{\draftspacing}{\let\CS=
\@currsize\renewcommand{\baselinestretch}{2.0}\tiny\CS}
\newcommand{\hugedraftspacing}{\let\CS=
\@currsize\renewcommand{\baselinestretch}{2.4}\tiny\CS}
\newcommand{\niceonespacing}{\let\CS=\@currsize\renewcommand{\baselinestretch}{1.1}\tiny\CS}
\newcommand{\nicetwospacing}{\let\CS=\@currsize\renewcommand{\baselinestretch}{1.2}\tiny\CS}
\newcommand{\nicethreespacing}{\let\CS=\@currsize\renewcommand{\baselinestretch}{1.3}\tiny\CS}
\newcommand{\singlespacingplusplus}{\let\CS=\@currsize\renewcommand{\baselinestretch}{1.35}\tiny\CS}
\newcommand{\nicefourspacing}{\let\CS=\@currsize\renewcommand{\baselinestretch}{1.4}\tiny\CS}
\newcommand{\nicefivespacing}{\let\CS=\@currsize\renewcommand{\baselinestretch}{1.5}\tiny\CS}
\newcommand{\nicesixspacing}{\let\CS=\@currsize\renewcommand{\baselinestretch}{1.6}\tiny\CS}
\newcommand{\nicesevenspacing}{\let\CS=\@currsize\renewcommand{\baselinestretch}{1.7}\tiny\CS}
\newcommand{\niceeightspacing}{\let\CS=\@currsize\renewcommand{\baselinestretch}{1.8}\tiny\CS}
\newcommand{\niceninespacing}{\let\CS=\@currsize\renewcommand{\baselinestretch}{1.9}\tiny\CS}
\makeatother%

\newcommand{\onlinesystemb}[1]{\mathrm{online}\hbox{-}\allowbreak\mathrm{{#1}}\hbox{-}\allowbreak\mathrm{Bribery}}
\newcommand{\onlinesystembk}[2]{\mathrm{online}\hbox{-}\allowbreak\mathrm{{#1}}\hbox{-}\allowbreak\mathrm{Bribery}[#2]}

\newcommand{\onlinesystemdb}[1]{\mathrm{online}\hbox{-}\allowbreak\mathrm{{#1}}\hbox{-}\allowbreak\mathrm{Destructive}\hbox{-}\allowbreak\mathrm{Bribery}}
\newcommand{\onlinesystemdbk}[2]{\mathrm{online}\hbox{-}\allowbreak\mathrm{{#1}}\hbox{-}\allowbreak\mathrm{Destructive}\hbox{-}\allowbreak\mathrm{Bribery}[k]}

\newcommand{\onlinesystemwb}[1]{\mathrm{online}\hbox{-}\allowbreak\mathrm{{#1}}\hbox{-}\allowbreak\mathrm{Weighted}\hbox{-}\allowbreak\mathrm{Bribery}}
\newcommand{\onlinesystemwbk}[2]{\mathrm{online}\hbox{-}\allowbreak\mathrm{{#1}}\hbox{-}\allowbreak\mathrm{Weighted}\hbox{-}\allowbreak\mathrm{Bribery}[#2]}

\newcommand{\onlinesystemdwb}[1]{\mathrm{online}\hbox{-}\allowbreak\mathrm{{#1}}\hbox{-}\allowbreak\mathrm{Destructive}\hbox{-}\allowbreak\mathrm{Weighted}\hbox{-}\allowbreak\mathrm{Bribery}}
\newcommand{\onlinesystemdwbk}[2]{\mathrm{online}\hbox{-}\allowbreak\mathrm{{#1}}\hbox{-}\allowbreak\mathrm{Destructive}\hbox{-}\allowbreak\mathrm{Weighted}\hbox{-}\allowbreak\mathrm{Bribery}[#2]}

\newcommand{\onlinesystempb}[1]{\mathrm{online}\hbox{-}\allowbreak\mathrm{{#1}}\hbox{-}\allowbreak\mathrm{\$Bribery}}
\newcommand{\onlinesystempbk}[2]{\mathrm{online}\hbox{-}\allowbreak\mathrm{{#1}}\hbox{-}\allowbreak\mathrm{\$Bribery}[#2]}

\newcommand{\onlinesystemdpb}[1]{\mathrm{online}\hbox{-}\allowbreak\mathrm{{#1}}\hbox{-}\allowbreak\mathrm{Destructive}\hbox{-}\allowbreak\mathrm{\$Bribery}}
\newcommand{\onlinesystemdpbk}[2]{\mathrm{online}\hbox{-}\allowbreak\mathrm{{#1}}\hbox{-}\allowbreak\mathrm{Destructive}\hbox{-}\allowbreak\mathrm{\$Bribery}[#2]}

\newcommand{\onlinesystempwb}[1]{\mathrm{online}\hbox{-}\allowbreak\mathrm{{#1}}\hbox{-}\allowbreak\mathrm{Weighted}\hbox{-}\allowbreak\mathrm{\$Bribery}}
\newcommand{\onlinesystempwbk}[2]{\mathrm{online}\hbox{-}\allowbreak\mathrm{{#1}}\hbox{-}\allowbreak\mathrm{Weighted}\hbox{-}\allowbreak\mathrm{\$Bribery}[#2]}

\newcommand{\onlinesystemdpwb}[1]{\mathrm{online}\hbox{-}\allowbreak\mathrm{{#1}}\hbox{-}\allowbreak\mathrm{Destructive}\hbox{-}\allowbreak\mathrm{Weighted}\hbox{-}\allowbreak\mathrm{\$Bribery}}
\newcommand{\onlinesystemdpwbk}[2]{\mathrm{online}\hbox{-}\allowbreak\mathrm{{#1}}\hbox{-}\allowbreak\mathrm{Destructive}\hbox{-}\allowbreak\mathrm{Weighted}\hbox{-}\allowbreak\mathrm{\$Bribery}[#2]}

\newcommand{\systemb}[1]{\mathrm{{#1}}\hbox{-}\allowbreak\mathrm{Bribery}}

\newcommand{\systemwb}[1]{\mathrm{{#1}}\hbox{-}\allowbreak\mathrm{Weighted}\hbox{-}\allowbreak\mathrm{Bribery}}

\newcommand{\systempb}[1]{\mathrm{{#1}}\hbox{-}\allowbreak\mathrm{\$Bribery}}

\newcommand{\systempwb}[1]{\mathrm{{#1}}\hbox{-}\allowbreak\mathrm{Weighted}\hbox{-}\allowbreak\mathrm{\$Bribery}}

\newcommand{\cale}{{\cal E}}
\newcommand{\calc}{{\cal C}}
\newcommand{\condition}{\mid}
\def\land{{\; \wedge \;}}
\newcommand{\littlep}{{p}}

\newcommand{\sigmatwo}{{\Sigma_2^{\littlep}}}

\newcommand{\deltatwo}{{\Delta_2^{\littlep}}}

  \newtheorem{theorem}{Theorem}[section]
  
  \newtheorem{proposition}[theorem]{Proposition}
  \newtheorem{fact}[theorem]{Fact}
  \newtheorem{definition}[theorem]{Definition}
\newtheorem{observation}[theorem]{Observation}

\newcommand{\manyone}{\ensuremath{\leq_{m}^{p}}}
\newcommand{\dttred}{\ensuremath{\leq_{dtt}^{p}}}

\newcommand{\p}{\ensuremath{\mathrm{P}}}
\newcommand{\np}{\ensuremath{\mathrm{NP}}}
\newcommand{\npnp}{\ensuremath{\np^{\np}}}

\newcommand{\dtime}{\ensuremath{\mathrm{DTIME}}}
\newcommand{\conp}{\ensuremath{\mathrm{coNP}}}

\newcommand{\pspace}{\ensuremath{\mathrm{PSPACE}}}

\newcommand{\OMIT}[1]{} %

\newcommand\sigmalevel[1]{\ensuremath{{\Sigma^p_{#1}}}}

\newcommand\pilevel[1]{\ensuremath{{\Pi^p_{#1}}}}

\newcommand{\qbf}{\ensuremath{\mathrm{QBF}}}

\newcommand{\systemucm}[1]{{\mathrm{{#1}}\hbox{-}\mathrm{UCM}}}

\newcommand{\systemwcm}[1]{{\mathrm{{#1}}\hbox{-}\mathrm{WCM}}}

\newcommand{\systemducm}[1]{{\mathrm{{#1}}\hbox{-}\mathrm{DUCM}}}

\newcommand{\systemdwcm}[1]{{\mathrm{{#1}}\hbox{-}\mathrm{DWCM}}}

\newcommand{\onlinesystemucm}[1]{{\mathrm{online}\hbox{-}\mathrm{{#1}}\hbox{-}\mathrm{UCM}}}

\newcommand{\onlinesystemwcm}[1]{{\mathrm{online}\hbox{-}\mathrm{{#1}}\hbox{-}\mathrm{WCM}}}

\newcommand{\onlinesystemducm}[1]{{\mathrm{online}\hbox{-}\mathrm{{#1}}\hbox{-}\mathrm{DUCM}}}

\newcommand{\onlinesystemdwcm}[1]{{\mathrm{online}\hbox{-}\mathrm{{#1}}\hbox{-}\mathrm{DWCM}}}

\newcounter{alg}

\newcounter{subalg}

\newcommand{\lahselfnote}[1]{}

\newcommand{\pair}[1]{\mathopen\langle{#1}\mathclose\rangle}

\title{The Complexity of Online Bribery\\in Sequential Elections}

\author{Edith Hemaspaandra\thanks{Supported in part
    by NSF grant DUE-1819546
    and
    a Renewed Research Stay grant from the
      Alexander von Humboldt Foundation.
Work done
    in part 
while 
  on sabbatical visits to ETH-Z\"urich and the University of
  D\"usseldorf.}\\Department of Computer Science\\
        Rochester Institute of Technology\\
        Rochester, NY 14623, USA 
\and
        Lane A. Hemaspaandra\thanks{Supported in part by
NSF grant CCF-2006496 and
      a Renewed Research Stay grant from the
      Alexander von Humboldt Foundation.
Work done
    in part 
while 
  on sabbatical visits to ETH-Z\"urich and the University of
  D\"usseldorf.}\\
        Department of Computer Science\\
        University of Rochester\\
        Rochester, NY 14627, USA
\and
J{\"o}rg Rothe\thanks{Supported in part by DFG grants RO~1202/14-2 and RO~1202/21-1.}\\
Institut f\"ur Informatik \\
        Heinrich-Heine-Universit{\"a}t D{\"u}sseldorf  \\
        40225 D\"usseldorf, Germany
      }
      \date{October 24, 2021}

\begin{document}
\sloppy

\maketitle

\begin{abstract}  

 Prior work on the complexity of bribery assumes that the bribery happens simultaneously, and that the briber has full knowledge of all votes.  However, in many real-world settings votes come in sequentially, and the briber may have a use-it-or-lose-it moment to decide whether to alter a given vote, and when making that decision the briber may not know what votes remaining voters will cast.

 We introduce a model for, and initiate the study of, bribery in such an online, sequential setting.  We show that even for election systems whose winner-determination problem is polynomial-time computable, an online, sequential setting may vastly increase the complexity of bribery, jumping the problem up to completeness for high levels of the polynomial hierarchy or even $\pspace$.  But we also show that for some natural, important election systems, such a dramatic complexity increase does not occur, and we pinpoint the complexity of their bribery problems.

 \end{abstract}

\section{Introduction}\label{sec:introduction}
In computational social choice theory,
the three most studied types of attacks on elections
are bribery, control, and
manipulation,
and the models of those that are studied seek to model the analogous
real-world actions.
Informally speaking, bribery means that an external agent, the
briber,
by bribing some
voters without exceeding a given budget seeks to influence the
outcome of an 
election; electoral control
refers to an external agent, commonly called the chair, who seeks to
influence the outcome of an election by structural changes such as
adding, deleting, or partitioning candidates or voters; and
manipulation (see Footnote~\ref{foo:ucm-wcm-ducm-dwcm}
for a formal definition) means
that some voter or coalition of voters may vote strategically.
These strategic actions and their applications in
artificial intelligence and multiagent systems have been
surveyed in many book chapters and articles~\cite{fal-hem-hem:j:cacm-survey,con-wal:b:barriers-to-manipulation-in-voting,fal-rot:b:handbook-comsoc-control-and-bribery,bau-rot:b:preference-aggregation-by-voting}.

Such studies are typically carried out for the model in which all the voters
vote simultaneously.  That sometimes is the case in the real world.
But it also is sometimes the case that the voters vote in
sequence---in what is sometimes called a roll-call election (see
Section~\ref{s:related} for some related work); in political
settings such as American political-party
presidential primary conventions, one may actually have a verbal roll
call for votes
going, for example, from Alabama to Alaska to Arizona
and so on through Wyoming.
That type of setting, i.e., sequential,
has been relatively recently introduced and
studied for control and manipulation---in particular, studies have
been done of both control and manipulation 
in the so-called online, sequential
setting~\cite{hem-hem-rot:j:online-manipulation,hem-hem-rot:j:online-candidate-control,hem-hem-rot:j:online-voter-control}.

In the present paper, we study the complexity of, and algorithms for,
the online, sequential case of bribery.
Briefly put, we are studying the case where the
voting order (and the voter weights and cost of bribing each voter)
is known ahead of time to the briber.  
But at the moment a voter seeks to vote, the voter's planned
vote is revealed to the briber, who then has a use-it-or-lose-it
chance to bribe the voter, by paying the voter's bribe-price
(and doing so allows that vote to be changed to any vote the briber desires).

The problem we are studying is the complexity of that decision.  In
particular, how hard is it to decide whether under optimal play on the part
of the briber there is an action for the
briber
regarding the current voter 
such that under continued future optimal play by the briber (in the
face of all future revelations of unknown information being
pessimal), the briber can reach a certain
goal (e.g., having one of
his or her two favorite candidates win; 
or not having any of his or her three most hated candidates win; those
two types of goals are examples of what are known respectively
as constructive
and destructive goals).
We mention that the text ``there is an action for the briber
regarding
the current voter such that under continued future optimal
  play by the briber (in the face of all future revelations of unknown
  information being pessimal)'' in the previous sentence
  is in fact, as will be made clear in Section~\ref{sec:online-brib-sequ},
about alternating existential and universal quantifiers.
Section~\ref{sec:online-brib-sequ}, and for some issues 
Section~\ref{s:related}, provide
a detailed discussion of
issues regarding the model,
the varying forms the costs in bribery can take (from actual
dollars to time or effort spent to risk accepted), and the fact
that, despite the typical associations with the word ``bribery,''
in many settings bribery is not modeling illegal, immoral, or evil acts.

The following list presents the section structure of our results.
\begin{enumerate}

\item Sections~\ref{sec:general-upper-bounds-unlimited}
  and~\ref{sec:general-upper-bounds-limited} establish our upper
  bounds---of $\pspace$ and the $\pilevel{2k+1}$ level of the
  polynomial hierarchy---on online
  bribery (i.e., online, sequential bribery,
  but we will for the rest of the paper and especially in our problem names 
  often omit the word ``sequential'' when the word ``online''
  is present)
  in the general case and in
  the case of being restricted to at most $k$ bribes.
  In Section~\ref{sec:discussion-of-upcoming-proofs}, we will briefly
  discuss why our
  $\pilevel{2k+1}$
  upper-bound proofs are far from trivial and how we
  meet their challenges through establishing 
  a new result, that may be of interest in its own right,
  about 
  alternating Turing machines whose accepting paths are
  ``weight-bounded.''
  
\item\label{p:why-cool}
  Section~\ref{sec:match-lower-bounds} proves that there are
  election systems, with simple winner problems, such that each of the
  abovementioned upper-bounds is tight, i.e., that
  PSPACE-completeness holds or $\pilevel{2k+1}$-completeness holds.
  Again, we will briefly discuss in
  Section~\ref{sec:discussion-of-upcoming-proofs} some substantial,
  novel challenges in our lower-bound proofs and how we surmount them.

\item{} 
In Section~\ref{sec:online-brib-spec},
we look at the complexity of online bribery
for various natural systems. We show that for both Plurality and
Approval, it holds that priced, weighted online bribery is $\np$-complete,
whereas all other problem variants of online bribery are in~$\p$.
Since these other problem variants in the case of
traditional (i.e., nonsequential)
bribery are
known to be $\np$-complete~\cite{fal-hem-hem:j:bribery},
this also shows that nonsequential bribery can be harder than
online bribery for natural systems.
In addition, we provide complete dichotomy theorems that distinguish
NP-hard from easy cases for
all our online bribery problems for scoring protocols
and additionally we show that Veto
elections, even with three candidates, have even higher lower
bounds for weighted online bribery, namely $\p^{\np[1]}$-hardness, where
$\p^{\np[1]}$ denotes the class of problems that can be solved by a
P machine querying its NP oracle at most once on each input.
\end{enumerate}
We handle weighted election systems throughout
this paper in the standard way that one would naturally expect.
In Appendix~\ref{sec:disc-weight-vers}, we discuss the
strengths and weaknesses of using this approach.

\section{Related Work}\label{s:related}
Our paper's general area is computational social choice, in which 
studying the complexity of election and preference aggregation
problems and manipulative attacks on them is a central theme.
There
are many
excellent surveys and book chapters on computational social
choice~\cite{bra-con-end:b:comsoc,rot:b:econ,bra-con-end-lan-pro:b:handook-of-comsoc},
and computational
social choice and computational complexity have a long history of
close, mutually beneficial interaction (see the survey~\cite{hem:c:bffs}).  

The prior papers most related to our work are the papers that
defined and studied the complexity of 
online control~\cite{hem-hem-rot:j:online-candidate-control,hem-hem-rot:j:online-voter-control},
of online manipulation~\cite{hem-hem-rot:j:online-manipulation}, and
of traditional (i.e., nonsequential)
bribery~\cite{fal-hem-hem:j:bribery}.  Particularly
important among those is online manipulation, as we will show
connections/inheritance between online manipulation and our problems.  We
also will show connections/inheritance between nonsequential 
manipulation and our problems.  Traditional
(i.e., nonsequential) manipulation was introduced by
Bartholdi, Tovey, and Trick~\cite{bar-tov-tri:j:manipulating} in the
unweighted case and by Conitzer, Sandholm, and
Lang~\cite{con-lan-san:j:when-hard-to-manipulate} in the weighted
case.

As alluded to in Section~\ref{sec:introduction}, in our sequential
problems the briber's goal in the so-called constructive case is, loosely
speaking, that at least
one of a collection of ``liked'' candidates be a winner, and in the
so-called destructive case is, loosely speaking,
that none of a set of ``disliked'' candidates is a
winner.  This approach to framing the goal
is the same as is used throughout the line of
papers mentioned above on online attacks on elections, and supports
connections and comparisons between this work and that earlier
work.  This model
differs from the single-candidate focused
model used in papers on nonsequential bribery,
and in Footnote~\ref{fn:goal-model} we will discuss
the model choice, and why we feel this goal is
natural, and will present some complexity inheritances
(and a noninheritance) between the
two approaches.  
But we mention here that this
approach to framing the goal, in addition to being the
settled one in papers on online electoral attack problems,
simply seems natural.  For example, agents often do come into an election
system focused not on a particular candidate winning, but having
a collection of candidates from which they hope that at least
one wins; and the agents often act accordingly.  And similarly, agents
often come into an election with some set of candidates they view
as so dangerous or wrong-headed that the agent wants to ensure that
none of them win.  More broadly, we view the problems
studied here
as natural and interesting.  Admittedly, the problems are
formal problems, and so are not perfectly capturing the noise of
reality.  However, formalizing and studying a crisp model of a 
problem is an important step, even if further steps are needed.
Relatedly, we mention that the worst-case nature of our analysis is
itself creating a quite hostile environment for the briber---and
so is modeling conservative bribers who want to handle the case
in which unrevealed things all come out against them---and
in our open questions collection at the end of this paper
we urge the study of online bribery in contexts
that go beyond the worst-case setting.

The existing work most closely related to our work on the effect on
alternating Turing machines
and formulas
of limits on existential actions is the work on online voter
control~\cite{hem-hem-rot:j:online-voter-control}, though the issues
tackled here are different and harder.

The work of Xia and Conitzer~\cite{con-xia:c:stackelberg-sequential}
(see 
also~\cite{slo:j:sequential-voting,dek-pic:j:sequential-voting-binary-elections,des-elk:c:sequential-voting,bab-dea-ten:c:sequential-voting-with-approvals})
that
defines and explores the Stackelberg voting game is also about
sequential voting, although unlike this paper
their analysis is game-theoretic
and is about manipulation rather than bribery.
Sequential (and related types of) voting have also been studied
in an axiomatic way~\cite{ten:c:transitive-voting}
and using Markov decision processes~\cite{par-pro:c:dynamic-social-choice},
though neither of those works focuses on issues of bribery.
Poole and Rosenthal~\cite{poo-ros:b:congress} provide a history
of roll-call voting.
This is used, for instance, in the US Senate;
Thomas discusses the strategic behavior of US senators who, seeking to
be re-elected, ``deliberately change the ideological tenor of their
roll-call voting during the course of their
terms''~\cite[p.~96]{tho:j:election-proximity-and-senatorial-roll-call-voting}.
Another real-life example of a roll-call election can be observed when in
a department meeting the chair goes around the table, asking each
faculty member---one after the other---about their preferences regarding
some important departmental matter (and perhaps about their reasons for these
preferences).

In the original paper on
nonsequential bribery there were other
types of bribery, e.g.,
bribery$'$,
unary-coding,
and succinct variants~\cite{fal-hem-hem:j:bribery}. 
Many other types have been studied since, e.g.,
microbribery~\cite{fal-hem-hem-rot:j:llull},
nonuniform
bribery~\cite{fal:c:nonuniform-bribery},
swap- (and its special case \mbox{shift-)~bribery}~\cite{elk-fal-sli:c:swap-bribery}
(see also~\cite{elk-fal:c:shift-bribery,bre-che-fal-nic-nie:j:prices-matter-shift-bribery,mau-nev-rot-sel:c:complexity-of-shift-bribery-in-iterative-elections}), and extension bribery~\cite{bau-fal-lan-rot:c:lazy}.
However,
for compactness and
since they are very natural, this paper focuses completely on
bribery in its eight
typical versions (as to prices,
weights, and constructive/destructive), except now in
an online, sequential setting.
It would be interesting to see in future work whether our model of
online bribery in sequential elections (to be formally described in
Section~\ref{sec:online-brib-sequ}) can also be applied to these
other variants of bribery.

Pulling back to the bigger picture, it is very important to stress
that bribery can be about bribing in the ``natural'' sense of the
word: paying people to change their votes.  But bribery more generally
models the situation where for each of a number of voters there is a
cost associated with changing that voter's vote.  The cost indeed
could be cash given to bribe them.  But it could also be the ``shoe
leather'' cost of sending campaign workers to the voters' doors to spend
the time to change
the voters' minds so that the voters
actually sincerely believe in, and thus vote in, a
given way.
A variant on this that is more explicitly sequential would
be a local political candidate canvassing door-to-door through a neighborhood
along a fixed path, before an
election,
and after an initial few moments of chatting forming an assessment of
the voter's preferences and then deciding whether to spend the time needed
to change the voter's mind.  
Or it could be that the cost is measuring the danger to
the briber of corrupting without the voter's knowledge the given vote
as it passes through the briber's hands (and the briber is operating
within a limit of how much total danger he or she is willing to
risk).

That is, bribery provides a relatively broad framework for allocating
a limited resource (framed as ``cost'') to change the votes of some among
a number of agents.  In fact, the original bribery paper of
Faliszewski, Hemaspaandra, Hemaspaandra~\cite{fal-hem-hem:j:bribery}
itself already allowed both prices and weights, and also studied the
case where the cost of the bribe varied based on ``how far'' from the
original preference of the voter the briber wanted to move the vote via
bribery (see that paper's coverage of so-called bribery$'$, and see
also the related notion of microbribery from Faliszewski, Hemaspaandra,
Hemaspaandra, and Rothe~\cite{fal-hem-hem-rot:j:llull});
and thus that paper itself was quite
flexible
in what its
framework encompassed.  That paper correctly
stressed that the ``bribery''
being modeled 
was not necessarily an illegal, immoral, or evil act
(\cite[p.~490]{fal-hem-hem:j:bribery}, see also
\cite[p.~280]{fal-hem-hem-rot:j:llull} and those papers,
from the
citation list later in this paragraph,
that are on
``campaign management'').
For example, the bribery could simply be a
transaction in some broader optimization seeking to find the lowest
cost to reach a certain type of outcome.
Papers since the work of 
Faliszewski, Hemaspaandra, Hemaspaandra~\cite{fal-hem-hem:j:bribery}
have proposed a wide range of new 
variants of 
the cost structure or the allowed bribery moves (or
even the vote types or vote ensembles),
depending on the situation being studied (as just a few
examples,~\cite{elk-fal-sli:c:swap-bribery,elk-fal:c:shift-bribery,bra-bri-hem-hem:j:single-peaked-bribery,fal-rei-rot-sch:j:manipulation-bribery-campaign-management-in-bucklin-fallback-voting,bre-fal-nie-tal:j:large-scale-election-campaigns,elk-fal-sch:j:campaign-management-approval-voting,mau-nev-rot-sel:c:complexity-of-shift-bribery-in-iterative-elections}).

This paper's
approach to the briber's goal, which is assuming worst-case
revelations of information, is inspired by the approach used in the
area known as online algorithms~\cite{bor-ely:b:online-algorithms}.
However, 
our goal notion is not of the competitive-ratio type
often used there, since here, as in general is true in computational
social choice, we are not dealing with numerically-valued notions of
degree of preference for a candidate.  Rather, as mentioned earlier
in this section, we adopt the goal model used by the existing line of
work on online attacks on elections.

Interesting work that is related---though somewhat distantly---in
flavor to our study is the paper of Chevaleyre et
al.~\cite{che-lan-mau-mon-xia:j:possible-winners-adding-welcome}
on the addition of candidates.  They also focus on the moment at which one
has to make a key decision, in their case whether all of a group of
potential additional candidates should be added.

\section{Preliminaries}\label{sec:preliminaries}

In this section, we first provide some basic notions from complexity
theory and social choice theory.  Then we formally define
and
discuss
our model of
online bribery.
Finally, we briefly discuss some
technical points of the proofs to come.
We mention that during a first reading of this paper,
the reader may wish to skip most or all of the footnotes, especially
the two long ones in Section~\ref{sec:online-brib-sequ}, since our footnotes
are mostly
used to give extra information, context, contrasts, and in some cases
claims.

\subsection{Basics}\label{sec:basics}
$\p$ is the class of decision problems in deterministic polynomial time.
$\np$ is the class of decision problems in nondeterministic polynomial time.
For each $k\geq 0$,
$\sigmalevel{k}$
is the class of decision problems in the $k$th $\Sigma$ level of
the polynomial hierarchy~\cite{mey-sto:c:reg-exp-needs-exp-space,sto:j:poly},
e.g., $\sigmalevel{0} = \p$,
$\sigmalevel{1} = \np$,
and $\sigmalevel{2} = \npnp$ (i.e., the class of sets accepted by
nondeterministic polynomial-time 
oracle Turing machines
given unit-cost access to an NP oracle).
For each $k\geq 0$,
$\pilevel{k} = \{ L \condition \overline{L} \in \sigmalevel{k}\}$, e.g., 
$\pilevel{0} = \sigmalevel{0} = \p$,
$\pilevel{1} = \conp$, and 
$\pilevel{3} = \conp^{\np^\np}$.
The polynomial-hierarchy level $\deltatwo = \p^{\np}$ is the class of sets 
accepted by deterministic polynomial-time oracle Turing machines given
unit-cost access to an $\np$ oracle, and $\p^{\np[1]}$ is the same class
restricted to one oracle query per input.
  Chandra, Kozen, and Stockmeyer~\cite{cha-koz-sto:j:alternation}
  defined alternating polynomial-time Turing machines and
  showed that the set of 
 languages accepted by 
 alternating polynomial-time Turing machines is
 exactly PSPACE (i.e., the class of problems
 that can be solved in polynomial space).  We
  will not go into detail
  about alternating Turing machines, but simply put, they are Turing
  machines that can make both universal and existential moves.

We say that $A \manyone B$ ($A$ polynomial-time many-one reduces to $B$)
exactly if there is a polynomial-time computable function $f$ such that
$(\forall x)[ x\in A \iff f(x) \in B]$.
\begin{fact}\label{f:m-closure}
For each complexity class $\calc
\in \{
\sigmalevel{0},\allowbreak
\sigmalevel{1},\allowbreak
\pilevel{1},\allowbreak
\sigmalevel{2},\allowbreak
\pilevel{2}, \dots\}$,
$\calc$ is closed downwards under 
polynomial-time many-one reductions, i.e.,
$(B  \in \calc \land A \manyone B) \implies A \in \calc$.
\end{fact}

Each of the classes mentioned in
Fact~\ref{f:m-closure} is even closed
downwards under what is known as polynomial-time
disjunctive truth-table reducibility~\cite{lad-lyn-sel:j:com}.
Disjunctive truth-table reducibility can be defined as follows.
We say that $A \dttred B$ ($A$ polynomial-time disjunctive truth-table
reduces to $B$)
exactly if there is a polynomial-time computable function $f$ such
that, for each $x$, it holds that (a)~$f(x)$ outputs a list of 0 or more
strings, and
(b)~$x \in A$ if and only if at least one string output by $f(x)$ is a
member of $B$.
(Polynomial-time many-one reductions are simply the special case of
polynomial-time disjunctive truth-table reductions where 
the polynomial-time disjunctive truth-table reduction's output-list function
is required to always
contain exactly one element.)
\begin{fact}\label{f:dtt-closure}
For each complexity class $\calc
\in \{
\sigmalevel{0},\allowbreak
\sigmalevel{1},\allowbreak
\pilevel{1},\allowbreak
\sigmalevel{2},\allowbreak
\pilevel{2}, \dots\}$,
$\calc$ is closed downwards under 
polynomial-time disjunctive truth-table reductions, i.e., 
$(B  \in \calc \land A \dttred B) \implies A \in \calc$.
\end{fact}
The above fact is obvious for P.  It is also 
easy to see and well known
for NP and $\conp$ (for example, the results follow
immediately from
the
result of Selman~\cite{sel:j:reductions-pselective} that NP
is closed downwards under so-called positive Turing
reductions).
The
results for the NP and coNP cases relativize (as Selman's mentioned
result's proof clearly relativizes).
That
gives
(namely, by relativizing the NP and coNP cases by complete sets
for NP, $\npnp$, etc.)\ the
claims for the higher levels of the polynomial hierarchy
(in fact, it gives something even stronger, since it gives downward
closure under disjunctive truth-table reductions that themselves
are relativized, but we won't need that stronger version in this paper).

All the many-one and disjunctive truth-table reductions discussed
in the paper will be polynomial-time ones.
So we henceforth will
sometimes skip the words ``polynomial-time'' when speaking of a polynomial-time
many-one or disjunctive truth-table reduction.

A set $L$ is said to be (polynomial-time many-one) hard for a class
$\calc$ (for short, ``$L$ is $\calc$-hard'')
exactly if $(\forall B \in \calc)[B
\manyone L]$.  If in addition $L \in \calc$, we say that 
$L$ is polynomial-time many-one complete for $\calc$, or simply
that $L$ is $\calc$-complete.

An (unweighted) election system $\cale$ takes as input
a voter collection $V$ and a candidate set $C$, such that
each element of $V$ contains a voter name
and a preference order
over the candidates in $C$;
and for us in this paper preference orders are always
total orders,
except when we are speaking of approval voting where
the preference orders are bit-vectors from $\{0,1\}^{\|C\|}$.
(For
  the rest of the preliminaries, we'll always speak of
  total orders as the preference orders' type, with it being implicit
  that when later in the paper we speak of
  and prove results about approval voting, all such places will
  tacitly be viewed as speaking of bit-vectors.)
The election system maps from that to
a (possibly nonproper) subset of $C$, often called the
winner set.
We often will call each element of $V$ a vote,
though as is common
sometimes we will use the term vote to refer just to the preference
order.
We often will use the variable names $\sigma$, $\sigma_1$, $\sigma_2$,~\dots,
$\sigma_i$ for total orders.
We allow election systems to, on some inputs, have no
winners.\footnote{Although
in social choice this is often disallowed, as has been 
discussed previously, see, e.g.,~\cite[Footnote~3]{fit-hem-hem:j:xthenx},
artificially excluding the case of
no winners
is unnatural, and many papers in computational
social choice allow this case.
A typical
real-world motivating example is that in Baseball Hall of Fame votes,
having no inductees in a given year is a natural outcome that
has at times occurred.}

For a given  (unweighted, simultaneous) election system, $\cale$, the
(unweighted) winner (aka the winner-determination) problem (in the
unweighted case) 
is the set $\{\pair{C,V,c} \condition c \in C \land$ $c$
is a winner of the election $(C,V)$ under election system~$\cale\}$.

For a given (\emph{unweighted}, simultaneous) election system, $\cale$,
the winner problem in the \emph{weighted} case will be the set of all
strings $\pair{C,V,c}$ such that $C$ is a candidate set, $V$ is a set
of weighted (via binary nonnegative integers as weights) votes (each
consisting of a voter name and a total order over $C$),
$c \in C$, and in the unweighted election created from this by
replacing each $w$-weighted vote in $V$ with $w$~unweighted copies of
that same vote, $c$~is a winner in that election under the
(unweighted) election system~$\cale$.
For an election system
$\cale$, it is clear that if the winner problem in the weighted
case is in $\p$, then so is the winner problem in the unweighted
case.  However, there are election systems $\cale$ for which the
converse fails.
The above approach to defining the weighted winner problem
is natural and appropriate for the election systems discussed
in this paper.
However, see Appendix~\ref{sec:disc-weight-vers}
for a discussion of the strengths and weaknesses of using this approach
to the weighted winner problem in other settings, and for more
discussion of the claims in this paragraph.

\subsection{Online Bribery in Sequential Elections}\label{sec:online-brib-sequ}

This paper is about the study of online bribery in sequential
elections.  In this setting, we are---this is the sequential
part---assuming that the voters vote in a well-known order,
sequentially, with each casting a ballot that expresses preferences
over all the candidates.  And we are assuming---this is the online
part---that the attacker, called ``the briber,'' as each new vote comes
in has his or her one and only chance to bribe that voter, i.e., to
alter that vote to any vote of the briber's choice.

Bribery has aspects of both the other standard
types of electoral attacks: bribery is like
manipulation in that one changes votes and it is like
(voter) control in that one is deciding on a set of voters (in
the case of bribery, which ones to bribe).  Reflecting
this, our model follows as closely as possible the relevant parts of
the existing models that study manipulation and control in online
settings~\cite{hem-hem-rot:j:online-manipulation,hem-hem-rot:j:online-voter-control,hem-hem-rot:j:online-candidate-control}.
In particular, we will follow insofar as possible both the model
of, and the notation of the model of, the
paper by Hemaspaandra, Hemaspaandra, and
Rothe~\cite{hem-hem-rot:j:online-voter-control} that
introduced the
study of online voter control in sequential elections. In
particular, we will follow the flavor of their model of control by
deleting voters, except here the key decision is not
whether to delete a given voter, but rather is whether a given voter
should be bribed, i.e., whether the voter's vote should be erased and
replaced with a vote supplied by the briber.  That ``replace[ment]''
part is more
similar to what happens in the study of online manipulation, which
was modeled and studied by Hemaspaandra, Hemaspaandra, and
Rothe~\cite{hem-hem-rot:j:online-manipulation}.  We will, as both
those papers do, focus on a key moment---a moment of
decision---and in particular on the complexity of deciding whether
there exists an action the briber can take, at that moment, such that
doing so will ensure, even under the most hostile of conditions regarding
the information that has not yet been revealed, that the briber 
will be able to
meet his or her goal.

If $u$ is a
voter and $C$ is a candidate set, an
\emph{election snapshot for $C$ and $u$} 
is specified by a triple $V =
(V_{<u}, u, V_{>u})$,
which loosely put (and in the paragraph after this one we will
provide
close coverage of the precise details) is 
made up of all voters in the order they vote, each accompanied
in models where there are prices and/or weights with their prices
and/or weights (which in this paper are assumed to
be nonnegative integers coded
in
binary).\footnote{\label{fn:binary}Why do we feel it
natural in most situations
for the prices
and weights
to be in binary rather than unary?  A TARK referee,
for example, asked whether it was not natural to assume that weights
and prices would always be small, or if not, would always be multiples
of some integer that when divided out would make the remaining numbers small.

Our
answer is that both prices and weights in many settings tend to be large,
and without any large, shared-by-all divisor.  To see this
clearly regarding weights, 
consider for example  
the number of shares of stock the various stockholders
hold in some large corporation or the
number of residents in each of the states of a country.  Prices 
too are potentially as rich and varied as are individuals and
objects, e.g., in some settings
each person's bribe-price might be the
exact fair market value of his or her house,
or might be closely related to the number of visits a web site they
own has had in the past year.

Pulling back, we note that requiring prices and weights to be in
unary is often tremendously (and arguably inappropriately)
\emph{helping} the algorithms as to
what their complexity is, since in effect one is ``padding'' the
many inputs' sizes as much as exponentially.  But if 
weighted votes are viewed as indivisible objects---and that is
indeed how they are typically
treated in the literature---the right approach
indeed is to code the weights in binary, and not to give
algorithm designers the potentially vastly lowered bar created by
the padding effect of coding the weights in unary.  Indeed,
it is known
in the study of nonsequential bribery that changing prices or weights
to unary can shift problems' complexities from
NP-hardness
to being in deterministic polynomial time~\cite[pp.~500--504]{fal-hem-hem:j:bribery}.

Also, people typically do code natural numbers
in binary, not unary.}
In addition, for each
voter voting before $u$ (namely, the voters in $V_{<u}$), also
included in this listing 
will be
the vote they cast (or if they were bribed, what vote was
cast for them by the briber) 
and whether they were bribed; and for $u$ the listing will also
include the
vote $u$ will cast unless bribed to cast a different vote.
So $V_{>u}$ is simply a list, in the order
they will vote, of the voters, if any, who come
after $u$, each also including the voter's price and/or weight data if we are
in a priced and/or weighted setting.  Further, the vote for $u$ and 
all the votes
in $V_{<u}$ must be votes over the candidate set $C$ (and in particular,
in this paper votes are total orderings of the candidates, e.g.,
$a > b> c$).

There is a slight overloading of notation above, in that we have not
explicitly listed in the structure the location of the mentioned extra
data.  In fact, our actual definition is that the first and last 
components of the 3-tuple $V$ are lists of tuples, and the middle
component is a single tuple.
Each of these contain the appropriate information, as mentioned above.
For example,  for priced, weighted bribery:
\begin{enumerate}
\item\label{item:1}
  the elements of the list $V_{<u}$ will be 5-tuples $(v_i,p_i,w_i,\sigma_i,b_i)$
whose components respectively are the voter's name, the voter's price,
the voter's weight, the voter's cast ballot (which is the voter's
original preference order if the voter was not bribed and is whatever
the voter was bribed into casting if the voter was bribed),
and
a bit specifying whether that vote resulted from being bribed,
\item\label{item:2} the middle component of
$V$ will be a tuple that contains the first four of those five components,
and
\item\label{item:3} the elements of the list $V_{>u}$ will contain the first three
  of the above-mentioned five components.
\end{enumerate}
Similarly, for example, for unpriced, unweighted bribery, the
three tuple types 
would respectively have three components, two components, and one component.

As a remaining tidbit of notational overloading, 
in some places we will speak of $u$ when we
in fact mean the voter name that is the first component of the
tuple that makes up the middle tuple of $V$.  That is, we will use $u$
both for a tuple that names $u$ and gives some of its properties,
and as a stand-in for the voter him- or herself.  Which
use we mean will always be clear from context.
(We mention in
passing that the fact that our voters and candidates have names
is consistent with the existing line of work on
online elections, see,
e.g.,~\cite{hem-hem-rot:j:online-manipulation,hem-hem-rot:j:online-voter-control}, though it also is in keeping with that fact that voters and
candidates typically truly do have names.  We will at times use this
in proofs.)

Let us, with the above in hand, define our notions of online bribery
for sequential elections.  Settings can independently allow or not
allow prices and weights, and so we have four basic types of bribery
in our online, sequential model, each having both
constructive and destructive versions.

Our specification of
these problems as languages is
centered around what Hemaspaandra, Hemaspaandra,
and Rothe~\cite{hem-hem-rot:j:online-manipulation} called a
\emph{magnifying-glass moment}.  This is a moment of decision as to
a particular voter.  To capture precisely what information the
briber does and does not have at that moment, and to thus allow us to
define our problems, we define a structure that we will call an
\emph{OBS}, which stands for \emph{online bribery setting}.  An
OBS is defined as a 5-tuple
$(C, V, \sigma, d, k)$, where $C$ is a set
of candidates;
$V = (V_{<u}, u,
V_{u<})$ is an election snapshot for $C$ and~$u$ as discussed
earlier; $\sigma$ is the
preference order of the briber;
$d \in C$ is a distinguished candidate;
and $k$ is a nonnegative integer (representing for unpriced cases the
maximum number of voters that can be bribed, and for priced cases
the maximum total cost, i.e., the sum of the prices 
of all the bribed voters).

Given an election system~$\cale$, we define the \emph{online unpriced, unweighted
  bribery problem}, abbreviated by
$\onlinesystemb{\cale}$, as the following decision problem.
The input is an OBS\@.
And the question is:
Does there exist a legal (i.e., not violating 
whatever bribe limit holds) 
choice by the briber
on
whether to bribe $u$ (recall that $u$ is specified
in the OBS,
namely, via 
the middle component of $V$) and, if the choice is
to bribe, of what vote to bribe $u$ into casting,
such that if the briber makes that choice 
then no matter what votes the remaining voters after
$u$ are (later) revealed
to have, the briber's goal (the meeting of which itself depends on
$\cale$ and will be defined explicitly two paragraphs from
now) can be reached by the current
decision regarding $u$ and by using the briber's future (legal-only, of
course) decisions 
(if any), each being made using the briber's then-in-hand
knowledge about what votes have been cast
by then?

Note that this approach is about alternating quantifiers.  It is
asking whether there is a current choice by the briber such that for all
potential revealed vote values for the next voter there exists a
choice by the briber such that for all potential revealed vote values
for the next-still voter there exists a choice by the briber such
that\ldots~and so on\ldots~such that the resulting winner set
under election system $\cale$ meets the briber's goal.
This is a bit more subtle than it might at first seem.  The briber is
acting
somewhat powerfully, since the briber is represented by existential
quantifiers.  But the briber is not all-powerful in this model.  In
particular, the briber can't see and act on future revelations of vote values;
after all, those are handled by a universal quantifier that occurs
downstream from an existential quantifier that commits the briber to
a particular choice.

In the above we have not defined what ``the briber's goal'' is, so
let us do that now.
$W_{\cale}(C,U)$
will denote
the winner set, according
to election system $\cale$,
of the (nonsequential) election $(C,U)$, where $C$ is
the candidate set and $U$ is the set of votes.
By \emph{the briber's goal} we mean, in the constructive case, that if
at the end of the above process $U'$ is the set of votes (some may
be the
original ones and some may be the result of bribes), it holds that
$W_{\cale}(C,U') \cap \{c \condition c \geq_{\sigma} d\} \neq
\emptyset$, i.e., the winner set includes some candidate (possibly
itself being~$d$) that the briber likes at least as much as the briber
likes $d$.  In the destructive case, the goal is to ensure that
no 
candidate that the briber hates as much or more than the briber hates
$d$ belongs to the winner set, i.e., the briber's goal is to ensure
that
$W_{\cale}(C,U') \cap \{c \condition d \geq_{\sigma} c\} =
\emptyset$.\footnote{\label{fn:goal-model}In each of these two cases, although we have in the
  problem statement used an order $\sigma$, that order $\sigma$
  really is merely being used to
  determine what set of candidates to (constructive case) try to get
  one of into the winner set, or to (destructive case) try to keep all
  of out of the winner set.  So one might wonder why we don't simply
  pass in such a set, rather than passing in an ordering.  The answer is that for
  the version of the problem that this paper is studying, namely the decision
  version, either way would be fine.
However, allowing an order to be
  passed in keeps our approach in harmony with the approach of earlier
  work~\cite{hem-hem-rot:j:online-manipulation,hem-hem-rot:j:online-voter-control,hem-hem-rot:j:online-candidate-control},
  which made that choice because if one studies the optimization
  version of the problem---e.g., in the constructive case trying to
  find the most preferred candidate within $\sigma$ that the briber can
  ensure will be a winner---a set does not bring in enough information
  to frame the problem but an order does.

  Our approach is also
  following---and thus better allowing comparisons and connections
  to---the earlier papers 
on online attacks on elections,
  in that our paper is using an upper segment
  of the order for the constructive case and a lower segment of the
  order for the destructive case (see~\cite[second paragraph of
  Footnote~4]{hem-hem-rot:j:online-voter-control} for more discussion
  of why that choice is most natural).
  That model
  intentionally
  does not limit one to---as is the case in the
  nonsequential case of bribery---specifying in
  the constructive case a single candidate whom one wants to win,
  or 
  in the destructive case a single candidate whom one wants to keep
  from being a winner.  There is so much uncertainty in the
  online setting
  that it seems natural to allow the
  attacker to have such flexible goals.
  Of course, the attacker is free to specify just one candidate, and
  so the just-one-focus-candidate cases are special cases of our model.
  And so, for each problem here, the version that is limited to 
  a single focus-candidate certainly many-one polynomial-time reduces
  to the same problem in our more-flexible-goal model.  So all 
  our upper-bound complexity results in our model immediately are 
  inherited by the one-focus-candidate sequential versions.  On the other
  hand, it follows from the present paper's work that,
  unless $\np = \conp$,
  there are cases when the one-focus-candidate sequential model
  falls into a simpler complexity class than the analogous
  problem in our model.  In particular, the second part of
  Theorem~\ref{t:veto} proves that
  $\onlinesystemwb{\mbox{\rm 3-candidate-Veto}}$ is
  $\p^{\np[1]}$-complete in our goal model,
  but the second bullet point of the proof of that theorem
  notes that in the one-focus-candidate sequential model the
  problem is NP-complete.  
}

In a moment, we will go on to define online bribery
variants that have prices and/or
weights.
However, let us first pause for a moment to give a
toy example
of $\onlinesystemb{\cale}$, simply to give some
concreteness
to the problem.
(Note: Later in the paper, Section~\ref{sec:online-brib-spec} will
build
efficient 
algorithms for many
important online bribery
problems.)
In our example here, we 
will not focus very much
on the actual formal structures that we have defined to
capture the problem, but will give the example informally.  In this
example, the election system will be 2-Approval, in which the
top two candidates in each voter's vote get one point each from that
voter and the other candidates get zero points from that voter.  And
whichever candidate (or candidates if there is a tie for most points)
gets the most points wins the election.  In our example, the candidate
set $C$ will be $a$, $b$, and $c$, the briber's
preference order $\sigma$ will be $a<b<c$, and
the briber's $d$ 
will be $b$.  So the briber's goal is to have at least one of
$b$ or $c$ be a winner.  Suppose that our magnifying-glass moment is
the moment when the second-to-last voter's vote is revealed, and
suppose that as we enter that moment---we will not focus on the
details of what came before---the current scores are 8 points for $a$,
and $7$ points each for $b$ and $c$.  Suppose that at this
magnifying-glass moment we are in the extreme case in which the
briber's budget has already been completely expended, and so the
briber is helpless to do any further bribing.  Let us completely
cover whether this instance is a Yes instance of 
$\onlinesystemb{\mbox{\rm 2-Approval}}$ (i.e., belongs to that set)
or is a No instance.  If the revealed vote of the second-to-last voter
is $a<b<c$ or is $a<c<b$---i.e., the points of that voter will go to
$b$ and $c$---then this is a Yes instance of the problem.  Why?  After
that second-to-last vote---and keep in mind that the briber has
no bribing budget
left---all three candidates would be tied at 8 points, and so since the final
voter has to give one point to exactly two candidates, whichever of
$b$ or $c$ gets one of those points will be a winner. So at least one
of $b$ or $c$ will be a winner, though we don't know which.  On the
other hand, if the revealed vote of the second-to-last voter is any of
the other four preference orders (which are the orders that will cause
that voter to give one of its points to $a$), then this is a No
instance of $\onlinesystemb{\mbox{\rm 2-Approval}}$.  That is so
simply because, since the briber has no bribing budget left, if the
second-to-last voter's points go to $a$ and $b$ (respectively, $a$
and $c$), then if the final voter gives his or her points to $a$
and $c$ (respectively, $a$ and $b$), then $a$ is the sole winner of
the election, and so the briber's goal has not been achieved.

Turning back to definitions, we already have defined both 
$\onlinesystemb{\cale}$ and 
$\onlinesystemdb{\cale}$.  Those both are in the
unpriced, 
unweighted 
setting.  And so
as per our definitions, the voters passed in as part of the problem statement
do not come with or need price or weight information.
In contrast, 
for the priced and/or weighted settings,
our definitions, naturally enough, require that 
the OBS includes those prices and/or weights.  And so the same
definition text that was used above defines all the other cases, except
that one must keep in mind
for the priced cases that when the ``bribery limit'' is mentioned one must
instead speak of the ``bribery budget,'' and in the weighted cases the
winner set $W$ is of course defined in terms of the weighted version
of the given voting system (which must, for that to be meaningful,
have a well-defined notion of what its weighted version is; 
Section~\ref{sec:basics} provides that notion for all systems 
in this paper, see also Appendix~\ref{sec:disc-weight-vers} for
further discussion).
Thus, we
also have tacitly defined the six problems
$\onlinesystempb{\cale}$,
$\onlinesystemdpb{\cale}$, $\onlinesystemwb{\cale}$,
$\onlinesystemdwb{\cale}$, $\onlinesystempwb{\cale}$, and
$\onlinesystemdpwb{\cale}$.
For an unpriced online bribery problem, we will postpend
the problem name with a ``[$k$]'' to define the version where
as part of the problem definition itself the bribery limit is---in
contrast with the above unpriced problem---not part of
the input
but rather is fixed to be the value $k$.
For example, $\onlinesystembk{\cale}{k}$ denotes the
unpriced, unweighted bribery problem where the number of voters who can
be bribed is set not by the problem input but rather is limited to be
at most $k$.

Why might these
``[$k$]'' cases be interesting and worthwhile to
study?  For many of our problems, we will show that
the versions without a bound on
the number of bribes are extraordinarily complex, namely,
PSPACE-complete. It is natural to wonder if such simpler versions
might drop to less extremely high levels of complexity,
and indeed we will see that that is the case.
So this is a type of parameterized investigation, though a somewhat
unusual one.
Also, 
in various situations one might naturally have both
a budget bound and a fixed upper bound on
the number of bribes that one can do.
For example, suppose a kingdom, to form each valuable alliance,
will need to always
offer in marriage a child of the monarch, plus some amount of treasure
that varies depending on the country being bribed into an alliance.
One might well suppose that the number of children of the monarch is
limited by some relatively small number, such as ten, and that
problem might
better model the problem if such bounds reduce the complexity, as we will
see that they seem to.  

Note that in each of the  ``[$k$]'' variants,
we tacitly are altering the definition of OBS from its standard
5-tuple,
$(C, V, \sigma, d, k)$,
to instead the 4-tuple
$(C, V, \sigma, d)$; that is because for these cases, the $k$ is
fixed as part of the general problem itself, rather than being a
variable part of
the individual instances.  For priced ``[$k$]'' variants,
there will be both a limit (being a variable part
of the input) on the total price of the bribes and a fixed as part
of the general problem itself
limit on the number of voters who can be bribed.

Of course, there are some immediate relationships that hold between
these eight problems.  One has to be slightly careful since there is a
technical hitch here.  We cannot for example simply claim that
$\onlinesystembk{\cale}{k}$ is a subcase of
$\onlinesystempbk{\cale}{k}$.  If we had implemented the unpriced case
by still including prices in the input but requiring them all to be 1,
then it would be a subcase.  But regarding both prices (weights), our
definitions simply omit them completely from problems that are not
about prices (weights).  In spirit, it is a subcase, but formally it
is not.  Nonetheless, we can still reflect the relationship between these
problems, namely, by stating how they are related via polynomial-time
many-one reductions. (We could even make claims regarding more restrictive
reduction types, but since this paper is concerned with complexity classes
that are
closed downwards under polynomial-time many-one
reductions, there is no
reason to do so.)  The following proposition (and the connections that follow from
it by the transitivity of polynomial-time many-one reductions) captures
this.

\begin{proposition}\label{p:reductions}
  \begin{enumerate}
    \item
For each $k>0$ and for each election system $\cale$,
  \begin{enumerate}
    \item 
      $\onlinesystembk{\cale}{k} \manyone \onlinesystempbk{\cale}{k}$,
    \item 
      $\onlinesystembk{\cale}{k} \manyone \onlinesystemwbk{\cale}{k}$,
\item $\onlinesystempbk{\cale}{k} \manyone\onlinesystempwbk{\cale}{k}$, and 
\item $\onlinesystemwbk{\cale}{k} \manyone\onlinesystempwbk{\cale}{k}$.
\end{enumerate}
\item The above item also holds for the case when all of its
  problems are changed to their destructive versions.
\item The above two items also hold for the case when
  all the ``$[k]$''s are removed (e.g., we have 
  $\onlinesystemb{\cale} \manyone \onlinesystempb{\cale}$).
\end{enumerate}  
\end{proposition}

One might ask why our model reveals the votes only as the voters vote,
but makes the voters' prices/weights known in advance.  The answer is
that in most natural situations of weighted voting, the weights are no
secret.  For example, it is well known how many votes each
state has at a political nominating convention.
And though people generally do not publicize how easily they can be
bribed, in our model for the priced cases, we are assuming that the
briber is acting with knowledge of each voter's price, perhaps due to
familiarity with the voter.  Similarly, we are taking the order that
voters vote as being known in advance, as also is the case in many
central situations, from roll-call votes by state at political
nominating conventions to the briber---a CS department
chairperson---meeting with faculty members one
at a time in his or her office to solicit (and perhaps bribe) their
votes on some central issue.
\lahselfnote{Stable for now and probably forever.  But just so this info is available for the future: We could move this paragraph earlier if the model is discussed earlier.  But for now, it is fine here, so leave it here.  That might change, but only if during the refereeing we have to do a restructuring, such as putting an example much earlier in the paper.  But having this discussion this late, though it has downsides, also has the big upside of the reader by now truly knowing much more and so being better able to appreciate and value such discussions.}

\subsection{A Brief Discussion of Some Upcoming Proofs}
\label{sec:discussion-of-upcoming-proofs}

Finally, let us mention and briefly discuss some of the novel proof
approaches needed to obtain our results.

The $\pilevel{2k+1}$ upper bounds that we will show in
Sections~\ref{sec:general-upper-bounds-unlimited}
and~\ref{sec:general-upper-bounds-limited} are far less
straightforward than upper bounds in the polynomial hierarchy
typically are.  Since bribes can occur on any voter (until one runs
out of allowed bribes), and so a yes-no decision has to be made, even
for the case of at most $k$ bribes, there can be long strings of
alternating existential and universal choices in the natural
alternating Turing machine programs for the problems.  And so there is
the threat that one can prove merely a PSPACE upper bound.

However, in Section~\ref{sec:result-about-altern} we prove a more
general result about alternating Turing machines that, while perhaps
making polynomially 
many alternations between existential and universal choices,
make most of
the existential choices
in a boring way (the exact restriction will be 
defined rigorously in that section).
Basically, we show that in the
relevant setting one can pull much of the existential guessing
upstream and make it external to the alternating Turing machine, and
indeed one can do so in such a way that one transforms the problem
into the disjunction of a polynomial number of uniformly generated
questions about actions of alternating Turing machines each of which
itself has at most $2k+1$ alternation blocks.  From that, we establish
the needed upper bound, both for the relevant abstract case of
alternating Turing machines and for our online bribery problems.

Regarding our result in Section~\ref{sec:match-lower-bounds} that
there are election systems, with simple winner problems, such that each
of the abovementioned upper-bounds is tight (so PSPACE-completeness or
$\pilevel{2k+1}$-completeness holds), there is
a substantial, novel challenge that the proof here has to overcome.
Namely, to prove for example $\pilevel{2k+1}$-hardness, we generally
need to reduce from quantified boolean formulas with particular
quantifiers applying to particular variables.  However, in online
bribery, the briber is allowed to choose where to do the bribing.
This in effect corresponds to having a formula with clusters of
quantified variables, yet such that we, as we attempt to prove theorems
related to these structures, don't have control over which quantifiers
are existential and which are universal.  Rather, in effect what the
online bribery setting will test is whether there exists an assignment
(consistent with the number of bribes allowed---which limits the
number of existential quantifiers one can set) of each quantifier to
be either existential or universal, such that for that
quantifier-assignment the formula evaluates as true.  (This is not at
all the same as quantifier exchange.  In quantifier exchange, the
exchanged quantifiers move around together with their associated
variables.)

However, we handle this by showing how to construct a new formula that
builds in protection against this setting.  In particular, we note
that one can take a quantified boolean formula and turn it into one
such that, in this Wild West setting of quantifier assignment, the new
formula can be made true by a legal (i.e., having at most as many
$\exists$ quantifiers as the original formula) quantifier assignment
exactly if the original formula is true.

\section{General Upper Bounds and Matching Lower Bounds}\label{sec:general-upper-bounds}
Even for election systems with simple winner problems, the best
general upper bounds that we can prove for our problems reflect an
extremely high level of complexity.

One might
wonder whether
that
merely is a weakness in our upper-bound
proofs.  However, in each case, we provide a matching completeness
result
proving that these really are the hardest problems in the
classes their upper bounds put them in.

However, in Section~\ref{sec:online-brib-spec}, we will see that for
many specific natural, important systems, the complexity is
tremendously lower than the upper bounds, despite the fact that the
present section shows that there exist systems that meet the upper
bounds.

\subsection{The General Upper Bound, Without Limits on the Number of Bribes}\label{sec:general-upper-bounds-unlimited}
This section covers upper bounds for
the case when any bribe limit/bribery budget is
passed in through the input---not hardwired into the problem
itself.
\begin{theorem}\label{t:pspace-general}
  \begin{enumerate}
  \item
    For each election system $\cale$ whose winner problem in the unweighted
  case is in
  polynomial time (or even in polynomial space),
each of the problems
  $\onlinesystemb{\cale}$,
  $\onlinesystemdb{\cale}$,
$\onlinesystempb{\cale}$, and 
$\onlinesystemdpb{\cale}$
is in $\pspace$.
\item
      For each election system $\cale$ whose winner problem in the weighted
  case is in
  polynomial time  (or even in polynomial space), each of the problems
$\onlinesystemwb{\cale}$,
$\onlinesystemdwb{\cale}$,
$\onlinesystempwb{\cale}$, and
$\onlinesystemdpwb{\cale}$
is in $\pspace$.
\end{enumerate}
\end{theorem}  

\begin{proof}
  Consider first the case in which the winner problem is in polynomial
  time.  For that problem, each of the eight problems can clearly be
  solved by (what is known as) an alternating polynomial-time Turing
  machine.
  It follows from the same paper of Chandra, Kozen, and
  Stockmeyer that defined alternating polynomial-time
  Turing machines~\cite{cha-koz-sto:j:alternation}
  that each is in $\pspace$, simply from the problems' definitions.

  One can see in various ways
  that the eight problems remain in $\pspace$ even if their winner problem
  is merely assumed to be in $\pspace$.  Perhaps the simplest way to see
  that is
  that it follows from the above combined with the fact that
  $\pspace^{\rm PSPACE} = \pspace$,
  in the model in which oracle queries
  are themselves polynomially length-bounded, since all eight of our problems
  when generalized to allowing $\pspace$ winner problems are,
  in the model in which oracle queries
  are polynomially length-bounded,
in 
  $\pspace^{\rm PSPACE}$.~\end{proof}

\subsection{The General Upper Bound, With Limits on the Number of Bribes}\label{sec:general-upper-bounds-limited}

Turning to the case where in the problem the number of bribes has
a fixed bound of $k$, these problems fall into the $\pilevel{2k+1}$
level of the polynomial hierarchy.  That is not immediately obvious.
After all, even when one can bribe at most $k$ times, one still for
each of the current and future voters seems to need to explore the
one-bit-per-voter decision of whether to bribe the voters (plus in
those cases where one does decide to bribe, one potentially has to
explore the exponential---in the number of candidates---possible votes
to which the voter can be bribed).  On its surface, for our problems,
that would seem to say that the number of alternations between
universal and existential moves that the natural polynomial-time
alternating Turing program for our problem would have to make is about
the number of voters (since for each voter we are asking whether 
there exists a legal bribing decision such that for all possible 
votes of the remaining voters, the briber's goal can be met)---a bound that would not leave the problem in any
fixed level of the polynomial hierarchy, but would merely seem to put
the problem in $\pspace$.

So these problems are cases where even obtaining the stated upper
bound is interesting and requires a twist to prove.  The twist is as follows.
On the surface the exploration of these problems has an unbounded
number of alternations between universal and existential states in
the natural, brute-force alternating Turing machine program.  But
for all but $k$ of the existential guesses on each accepting
path, the guess is a boring one,
namely, we guess that regarding that voter we don't bribe.
We will show, by proving a more general result about alternating
Turing machines and restrictions on the structure of their maximal
existential move segments (i.e., maximal sequences of existential guesses) along accepting paths, a 
$\pilevel{2k+1}$ upper bound on our sets of interest.
In some sense, in terms of being charged
as to levels of the polynomial hierarchy, we will be showing
that if for
a certain collection of 0-or-1 existential
decisions one on each accepting path chooses 0 all but a fixed
number of times (although for the other times one
may then make many more nondeterministic
choices), one can manage to in effect
not be charged at all for the guessing acts
that guessed~0.

We know of only one result in the literature that is anything like
this.  That result, which also came up in the complexity of online
attacks on elections, is a result of Hemaspaandra, Hemaspaandra, and
Rothe~\cite{hem-hem-rot:j:online-voter-control}, where in the context
not of bribery but of voter control they showed that for
each fixed $k>0$ it holds that, for each polynomial-time alternating Turing
machine $M$ whose alternation blocks are each one bit long and that for at
most $k$ of the existential blocks guess a zero, the language accepted
by $M$ is in $\conp$.

In contrast, in the present
paper's case we are in a far more complicated situation, since in bribery our
existential blocks are burdened not just by 1-bit bribe-or-not
decisions, but for the cases when we decide to try bribing, we need to
existentially guess what bribe to do.  And so we do not stay in
$\conp$ regardless of how large
 $k$ is, as held in that earlier case.
 But we show that we can at least limit the growth to at most $2k+1$
 alternating quantifiers---in particular, to the class $\pilevel{2k+1}$.
And since we
later provide problems of this sort that are complete for
$\pilevel{2k+1}$, our $2k+1$ is optimal unless the polynomial
hierarchy collapses.

We will approach this in two steps.  First,
as Section~\ref{sec:result-about-altern},
we will prove the result
about alternating Turing machines.  And then,
as Section~\ref{sec:upper-bound-results},
we will apply that to
online bribery in the case of
only globally fixed numbers of bribers being allowed.

\subsubsection{A Result about Alternating Turing Machines}\label{sec:result-about-altern}

Briefly put, an alternating Turing machine~\cite{cha-koz-sto:j:alternation}
(aka~ATM) is a generalization of
nondeterministic and conondeterministic computation.
We will now briefly review the basics
(see~\cite{cha-koz-sto:j:alternation} for a more complete
treatment).
An ATM can make both universal and existential choices.  For a
universal ``node'' of the machine's action to evaluate to true, all
its child nodes (one each for each of its possible choices) must
evaluate to true.  For an existential ``node'' of the machine's action
to evaluate to true, at least one of its child nodes (it has one child
node for each of its possible choices) must evaluate to true.  A
leaf of the computation tree (a path, at its end) is said to
evaluate to true if the path halted in an accepting state and is
said to evaluate to false if the path halted in a nonaccepting state.
(As our machines are time-bounded, all paths halt.)  Without loss of
generality,
in this paper we assume that each universal or existential node has
either two children (namely, does a universal or existential split
over the choices 0 and 1; we will often call this a 
``1-bit move'') or has exactly one child (it does a
trivial/degenerate
universal or existential choice of an element from the one-element set
$\{0\}$; we will often call this a ``0-bit move'').  
The latter case is in effect a deterministic move, except
allowing degenerate $\forall$ steps of that sort will let us put a
``separator'' between otherwise contiguous $\exists$ computation
segments.
Of course, long existential guesses can be done in this model,
for example by guessing a number of bits sequentially.
An
ATM accepts or rejects based on what its root
node evaluates to (which is determined inductively in the way described
above).

\begin{definition}\label{d:weight}
  Consider a path $\rho$ in the tree of an ATM\@.  The \emph{weight}
  of that path is as follows.  Consider all maximal segments of
  existential nodes with their guesses along the path.  (As mentioned above, we may
  without loss of generality, and do, assume that each nonleaf node is
  $\exists$ or $\forall$, although perhaps a degenerate such node
  in the way mentioned above.)
  The weight of path $\rho$
  is its number of maximal existential segments such that the concatenation
  of the bits guessed in that segment is not the 1-bit string 0 (i.e., the number
of maximal existential segments of length at least 2 plus the number of
maximal existential segments of length 1 with guess 1).
\end{definition}

Let us illustrate this, as Figure~\ref{f:left}.  In the figure,
the illustrated path (the leftmost one at the left edge of the tree)
has weight 0; it has three maximal existential segments, but each
is of length one and makes the guess 0. If we change the last $\forall$ to $\exists$
we have a path of weight 1; it has two maximal existential segments, 
the first one is of length 1 and makes the guess 0, and so does not count, but the
second one is of length 3, and so does count.
\begin{figure}[!tbp]
   \ctikzfig{fig-left-path}
      \caption{\label{f:left}A weight 0 path in the tree of an
        ATM\@.}
\end{figure}

With this definition in hand, we can now state our key theorem showing
that limited weight on accepting paths for ATMs simplifies the complexity of the
languages accepted.  The result is one where one may go back and forth
between thinking it is obvious and thinking it is not obvious.  In
particular, note that even on accepting paths of weight
at most $k$, it is completely possible that the number of alternations
between existential and universal nodes may be far greater than $k$
and may be far greater than $2k+1$, and indeed may grow unboundedly as
the input's size increases (and this might naturally lead one to
worry that perhaps even 
a machine with
bounded accepting-path weight 
potentially might 
accept a $\pspace$-complete language).
What the theorem below is saying is that
despite that, machines with bounded weight on their accepting paths
still accept only $\pilevel{2k+1}$ sets.  

\begin{theorem}\label{t:2k+1}
  Let $k\geq 0$ be fixed.  Each polynomial-time ATM $M$ such that
  on no input does 
  $M$ have an accepting path of weight strictly
  greater than $k$
  accepts a
  language in $\pilevel{2k+1}$.
\end{theorem}

\begin{proof}
  Let $k\geq 0$ be fixed.  Let $L$ be the language accepted by
  polynomial-time ATM $M$ that has the property that each of its
  accepting paths has weight at most $k$.  Our goal is to prove that
  $L \in \pilevel{2k+1}$.  

  We will do so by proving that there is a set $G \in \pilevel{2k+1}$
  such that $L \dttred G$, i.e., $L$ polynomial-time disjunctively
  truth-table reduces to $G$.  By Fact~\ref{f:dtt-closure},
  it follows that $L \in \pilevel{2k+1}$.

  Let $q$ be a nondecreasing polynomial that upper-bounds the running time
  of $M$.  Let $\pair{\cdot,\cdot}$ be a standard pairing function, i.e.,
  a polynomial-time computable, polynomial-time invertible bijection between
  $\sigmastar \times \sigmastar$ and $\sigmastar$.
  Recall that every step of our ATM $M$ involves either a 0-bit existential
  move
  (which we'll think of basically as existentially choosing 0 from the
  choice palette set $\{0\}$)
  or a 1-bit existential move (which, recall, involves choosing
  one element from the choice palette set $\{0,1\}$ with the machine
  enforcing an ``or'' over the two children thus reached)
   or a 0-bit universal move
  (which we'll think of basically as universally choosing 0 from the
  choice palette set $\{0\}$)
  or a 1-bit universal move.  And the 1-bit moves involve successor
  states hinged on whether the move-choice is a 0 or a 1.
 (As
  Turing machines are standardly defined, there can be (one or multiple)
  successor states to a given state, hinged on a (degenerate or
  nondegenerate) choice.)

  $G$ will be the set of all $\pair{x,s}$ such that all of the 
  list of conditions that we will give below hold relative to $x$ and $s$.
  The intuition here is that $s$ is a
  bit-vector whose $i$th bit controls how the $i$th maximal
  existential segment is handled.  In particular,
  if that $i$th bit is a~1, then the segment moves
  forward unrestrained.  But if that $i$th bit is a~0, then we expect and
  require (and cut off that part of the tree otherwise) the maximal
  existential segment to be a single existential step (either guessing
  a bit from $\{0,1\}$ or the allowed but superfluous existential step
  of guessing a bit
  from the one-element set $\{0\}$) and we 
  basically will (as described
 in part of step~\ref{step:cut} below)
    cut that step out of the tree by replacing it by a trivially 
  universal step.
  Returning to our defining of $G$,
  the set $G$ will be all $\pair{x,s}$ such that all 
  the following claims hold.
  \begin{enumerate}
  \item $x \in \sigmastar$ and $s \in \{0,1\}^{q(|x|)}$.

  \item The number of ``1''s in the bit-string $s$ is at most $k$.
  \item\label{step:cut} $M$ accepts when we simulate it on input $x$ but with the
    following changes in the machine's action.  

    As one simulates $M$ on a given path, consider the first
    existential node (if any) that one encounters.

    If the first bit of $s$ is
    a 1, then for that node we will directly simulate it, and
    on all paths that follow from this one, on all the following
    existential nodes (if any) that are in an unbroken segment of
    existential nodes from this one, we will similarly directly
    simulate them.  On the other hand, if the first bit of $s$ is a~0,
    then (a)~if it is the case that if the current node
    makes the choice 0 then the node that follows it is an existential
    node, then the current path halts and rejects (because something that
    $s$ is specifying as being
    a maximal existential segment consisting of a single~0
    clearly is not); and (b)~if~(a) does not hold (and so the
    node that follows if we make the choice~0 is either
    universal or a leaf), then do not take
    an existential action at the current node but rather implement it
    as a degenerate universal step (namely, a ``$\forall$'' guess over
    one option, namely, 0, matching as to next state and so on
    whatever the existential node would have done on the choice
    of~0).

    If the path we are simulating didn't already end or get cut off
    during the above-described handling of its first, if any,
    maximal existential segment, then continue on until we hit the
    start of its second existential segment.  We handle that exactly
    as described above, except now our actions are controlled not by
    the first bit of $s$ but by the second.

    And similarly for the third maximal existential segment, the fourth,
    and so on.

    All other aspects of this simulation are unchanged from $M$'s own
    native behavior.
  \end{enumerate}

  Note that $G\in\pilevel{2k+1}$.  Why?  Even in the worst of cases for us,
  the modified computation of $M$ starts with a $\forall$ block and then has $k$
  $\exists$ blocks each separated by a $\forall$ block; and then we
  finish with a $\forall$ block.  But then
  our ATM as it does that simulation starts with $\forall$
  and has $2k$ alternations of quantifier type, and thus has $2k+1$ alternation
  blocks with $\forall$ as the leading one.  And so by
  Chandra, Kozen, and Stockmeyer's~\cite{cha-koz-sto:j:alternation}
  characterization of the languages accepted by ATMs with that leading
  quantifier and that number of alternations, this set is in
  \pilevel{2k+1}.

  Finally, we argue that $L \dttred G$.  In particular, note that 
  $L \dttred G$ via the reduction that on input $x$ generates
  every length $q(|x|)$ bit-string having less than or equal to $k$
  occurrences of the bit~1, and as its $\dttred$ output list outputs
  each of those paired with $x$.  This list is easily generated
  and is polynomial in size.  In particular, the number of
  pairs in the list is clearly at most
  $\sum_{0\leq j \leq k} {q(|x|) \choose j}
  \leq \sum_{0\leq j \leq k}  (q(|x|))^j \leq (k+1) (q(|x|))^k$.
  That completes the proof.~\end{proof}

In the above, we focused on maximal existential guess sequences, and
limiting the number of those, on accepting paths, whose
bit-sequence-guessed was other than the string 0.  So we barred from
accepting paths any maximal existential guess sequences that contain a
1 and any that have two or more bits.  We mention in passing that we
could have framed things more generally in various ways.  For example,
we could have made each maximal existential guess be of a fixed
polynomial length and could have defined our notion of a ``boring''
guess sequence not as the string ``0'' but as a string of 0's of
exactly that length.  The $\pilevel{2k+1}$ upper bound holds also in
that setting, via only slight modifications to the proof.

\subsubsection{The Upper-Bound Results Obtained by Applying 
  the Previous Section's Result about Alternating Turing Machines%
}\label{sec:upper-bound-results}
We now establish our upper bounds on online bribery.

\begin{theorem}\label{t:parameterized-case-upper-bounds}
  \begin{enumerate}
  \item\label{p:k-uw}
For each $k \in \{0,1,2,\dots\}$, and for 
each election system $\cale$ whose winner problem in the unweighted
  case is in
  polynomial time,\footnote{\label{fn:1}Unlike
    Theorem~\ref{t:pspace-general}, we cannot allow the winner problem
    here to be in $\pspace$ and argue that the rest of the theorem
    holds unchanged.  However, we can allow the winner problem here to
    even be in $\np\cap\conp$, and then the rest of the theorem holds
    unchanged.  The key point to notice to see that that holds is---as
    follows immediately from the
    fact that
    $\np^{\np\cap\conp} = \np$~\cite{sch:j:low}---that
    for each $k \geq 0$, 
    $\np \cap \conp$ is $\pilevel{2k+1}$-low, i.e.,
    that $\left({\pilevel{2k+1}}\right)^{\np\cap\conp} =  \pilevel{2k+1}$.}
 each of the problems
  $\onlinesystembk{\cale}{k}$,
  $\onlinesystemdbk{\cale}{k}$,
$\onlinesystempbk{\cale}{k}$, and 
$\onlinesystemdpbk{\cale}{k}$
is in $\pilevel{2k+1}$.
\item\label{p:k-w}
For each $k \in \{0,1,2,\ldots\}$, and for 
each election system $\cale$ whose winner problem in the weighted
  case is in
  polynomial time,\footnote{As in the case of Footnote~\ref{fn:1},
    the rest of the theorem remains unchanged even if we relax
    the ``polynomial time'' to instead be ``$\np\cap\conp$.''}
each of the problems
$\onlinesystemwbk{\cale}{k}$,
$\onlinesystemdwbk{\cale}{k}$,
$\onlinesystempwbk{\cale}{k}$, and
$\onlinesystemdpwbk{\cale}{k}$
is in $\pilevel{2k+1}$.
\end{enumerate}
\end{theorem}  

\begin{proof}
Let $k\geq 0$ be fixed.
Let us start by arguing that 
$\onlinesystempwbk{\cale}{k} \in 
\pilevel{2k+1}$. After that, we will show 
the remaining cases of part~\ref{p:k-w} and then
we will show part~\ref{p:k-uw}.

As noted (for the case without the bound of~$k$)
in Section~\ref{sec:online-brib-sequ}, what is really going on here
is about alternating quantifiers.  Consider a given input to the problem
$\onlinesystempwbk{\cale}{k}$.
Let the voter under consideration (i.e., $u$) in the focus moment of
that problem
just for
this paragraph be referred to as  $u_1$, and let the ones coming
after it be called, in the order they occur, $u_2$, $u_3$,~$\dots$, $u_\ell$. 
What the membership problem is in essence asking is whether
there \emph{exists} an allowable (within both the price budget and
the global limit of $k$ allowed bribes)
choice as to whether to bribe $u_1$ (and if the decision is to bribe,
then whether there  \emph{exists} a vote to which to bribe $u_1$)
such that, \emph{for each} vote that $u_2$ may then be revealed to have,
there 
\emph{exists} an allowable (within both the price budget and
the global limit of $k$ allowed bribes)
choice as to whether to bribe $u_2$ (and if the decision is to bribe,
then whether there  \emph{exists} a vote to which to bribe $u_2$)
such that, \emph{for each} vote that $u_3$ may then be revealed to
have,~$\dots$,~such that
there \emph{exists} an allowable (within both the price budget and
the global limit of $k$ allowed bribes)
choice as to whether to bribe $u_\ell$ (and if the decision is to bribe,
then whether there  \emph{exists} a vote to which to bribe $u_\ell$)
such that
$W_{\cale}(C,U') \cap \{c \condition c \geq_{\sigma} d\} \neq
\emptyset$ (recall that that inequality says that 
the winner set includes some candidate
that the briber likes at least as much as the briber
likes $d$; $U'$ is here representing the vote set after all
the voting/bribing, as per Section~\ref{sec:online-brib-sequ}'s definitions).

Note that for at most $k$ of the choice blocks associated with
$u_1, u_2, \dots$ can we make the choice to bribe.
 (In fact,
if we have already done bribing of one or more
voters in $V_{< u}$, then our remaining number of allowed
bribes will be less than $k$.)  Keeping that in mind, imagine
implementing the above paragraph's alternating-quantifier-based
algorithm on a polynomial-time ATM\@.  In our model, every step is either
a universal or an existential one, and let us
program up all deterministic computations that are part of the above
via degenerate universal steps.  (We do that rather than
using degenerate existential steps since those degenerate
existential steps would interact fatally with our definition
of maximal existential sequence; we really need those places
where one guesses that one will not bribe to be captured as a
maximal existential sequence of length one with guess bit~0; this
comment is quietly using the fact that when making a 1-bit choice
as to whether to bribe we associate the choice 1 with ``yes bribe
this voter'' and 0 with ``do not bribe this voter.'')
In light of that and the fact that we know that
the weighted winner problem of election system $\cale$ is in~$\p$, the
limit of~$k$
ensures that no accepting path will have weight greater than $k$.
And so by Theorem~\ref{t:2k+1}, we have that
$\onlinesystempwbk{\cale}{k} \in 
\pilevel{2k+1}$.

By the exact same argument,
except changing the test at the end to
$W_{\cale}(C,U') \cap \{c \condition d \geq_{\sigma} c\} =
\emptyset$,
we have that 
$\onlinesystemdpwbk{\cale}{k} \in 
\pilevel{2k+1}$.

From these two results, it follows by Proposition~\ref{p:reductions}
that $\onlinesystemwbk{\cale}{k} \in \pilevel{2k+1}$ and
$\onlinesystemdwbk{\cale}{k} \in \pilevel{2k+1}$.

That completes the proof of part~\ref{p:k-w} of the theorem.
Now, we cannot simply invoke
Proposition~\ref{p:reductions}
to claim that part~\ref{p:k-uw} holds.  The reason is that
part~\ref{p:k-uw}'s hypothesis about the winner problem merely
puts the unweighted winner problem in $\p$, but 
the proof we just gave of part~\ref{p:k-w} used the fact
that for that part we could assume that the
weighted winner problem is in~$\p$.

However, the entire construction of this proof works perfectly
well in the unweighted case, namely, we are only given that
the unweighted winner problem is in $\p$, but the four problems
we are studying are the four unweighted problems of
part~\ref{p:k-uw} of the theorem statement.  So we have that
 each of the problems
  $\onlinesystembk{\cale}{k}$,
  $\onlinesystemdbk{\cale}{k}$,
$\onlinesystempbk{\cale}{k}$, and 
$\onlinesystemdpbk{\cale}{k}$
is in $\pilevel{2k+1}$.  That completes the proof of the theorem.%
{}~\end{proof}

\subsection{Matching Lower Bounds}\label{sec:match-lower-bounds}
For each of the $\pspace$ and $\pilevel{2k+1}$ upper bounds
established so far in this section, we can 
in fact establish a matching lower bound.  We show
that by, for each, proving that there is
an election system, with a polynomial-time winner
problem, such that the given problem is polynomial-time
many-one hard for the relevant
class (and so, in light of the upper-bound results,
is polynomial-time many-one complete for the relevant class).

\begin{theorem}\label{t:lower}
  \begin{enumerate}
  \item\label{p:lower-unbounded}
    For each problem $I$ from this list of problems:  $\onlinesystemb{\cale}$,
$\onlinesystemdb{\cale}$,
$\onlinesystempb{\cale}$, 
$\onlinesystemdpb{\cale}$,
$\onlinesystemwb{\cale}$,
$\onlinesystemdwb{\cale}$,
$\onlinesystempwb{\cale}$, and
$\onlinesystemdpwb{\cale}$,
there exists an (unweighted)  election system $\cale$, whose
winner problem in both the unweighted case 
and the weighted  case is in
polynomial time,
such that $I$ is $\pspace$-complete.
\item\label{p:lower-2}
For each $k \in \{0,1,2,\ldots\}$, and for 
each problem $I$ from this list of problems:
    $\onlinesystembk{\cale}{k}$,
  $\onlinesystemdbk{\cale}{k}$,
$\onlinesystempbk{\cale}{k}$, 
$\onlinesystemdpbk{\cale}{k}$,
$\onlinesystemwbk{\cale}{k}$,
$\onlinesystemdwbk{\cale}{k}$,
$\onlinesystempwbk{\cale}{k}$, and
$\onlinesystemdpwbk{\cale}{k}$,
there exists an (unweighted)  election system $\cale$, whose
winner problem in both the unweighted case 
and the weighted  case is in
polynomial time,
such that $I$ is $\pilevel{2k+1}$-complete.
\end{enumerate}
\end{theorem}

\begin{proof}
  All sixteen upper bounds follow from
  Theorems~\ref{t:pspace-general}
  or~\ref{t:parameterized-case-upper-bounds}.  So we need only prove the matching lower bound in
  each case.

We focus on the problem
  $\onlinesystembk{\cale}{k}$ and will discuss the other fifteen cases 
in Section~\ref{s:fifteen}. 
   We will show that there is an election system $\cale$, whose
  unweighted and weighted winner problems are in polynomial time,
  for which this problem is $\pilevel{2k+1}$-hard.

  To do this, we will insofar as possible be inspired by 
    Hemaspaandra, Hemaspaandra, and Rothe's~\cite{hem-hem-rot:j:online-manipulation}
    work on online manipulation, and will insofar
  as possible adopt its notations and phrasings.  This will
  help the novel challenge here to better stand out.  That novel
  challenge is the
  one mentioned in list item~\ref{p:why-cool} on
  page~\pageref{p:why-cool} of the current paper.  In the manipulation case, we can by
  shaping the output of the reduction (from a
  $\pilevel{2k+1}$-complete set) control which voters are
  manipulators, and so can match and control the quantifier structure
  of the quantified boolean formula whose truth we are trying to
  determine via an online manipulation question.  However, in online
  bribery, the briber can bribe any voters the briber wants (within the
  other constraints).  This gives the briber quite a lot of freedom---enough
  to in effect shuffle which quantifiers are which in the boolean formula
  we are trying to in effect simulate; but that would cause chaos and would
  break our simulation's connection to the formula.

  We will handle this
  by doing our simulation in effect on a transformed version of the formula
  that is immune to quantifier shuffling, thanks to the transformed
  formula being
  rigged so that any assignment as to which quantifiers are chosen to be
  $\exists$ and which are chosen to be $\forall$ cannot possibly lead
  to a successful bribe unless the assignment is precisely the one
  that causes the online bribery to in effect correctly model and study the
  underlying formula we care about as to where the quantifiers are.
  (Put somewhat differently, we
  create an online bribery problem where we in effect are forcing the
  hand of the briber as to whom to bribe.)

  We will show that this can be done quite directly, due to the power of
  boolean formulas.  Let us first show how to do this, and then we
  will use this in our proof.

\subsubsection{Transforming the Formula}

  Consider any quantified boolean
  formula, i.e.,
  $(Q_1 \overrightarrow{x_1})
  (Q_2 \overrightarrow{x_2})
  \cdots
  (Q_\ell \overrightarrow{x_\ell}) [ F(\overrightarrow{x_1},
  \overrightarrow{x_2},
  \dots,\allowbreak
  \overrightarrow{x_\ell})]$, where each $Q_i$ is either an $\exists$
  or a $\forall$, each $\overrightarrow{x_i}$ is shorthand for
  a list of boolean variables (with no boolean variable appearing
  in more than one of the $\overrightarrow{x_i}$'s), and $F$ is a
  propositional boolean formula.  Let $j$ be
  the number of $Q_i$ that are $\exists$.
  Now, suppose
  we want to put that formula into a setting where that same interior
  formula $F$ is used, but the $\ell$ quantifiers can be 
  reassigned in every possible way that has at most $j$ of them
  being $\exists$; and then we do an ``OR'' over all the thus-generated
  quantified boolean formulas.  For example, if our original formula
  is 
  \[
    (\forall \overrightarrow{x_1})
  (\exists \overrightarrow{x_2})
  (\forall \overrightarrow{x_3}) [ F(\overrightarrow{x_1},
  \overrightarrow{x_2},
  \overrightarrow{x_3})],
\]
then what we will in fact be assessing in that setting is this statement:
\begin{gather*}
    (\exists \overrightarrow{x_1})
  (\forall \overrightarrow{x_2})
  (\forall \overrightarrow{x_3}) [ F(\overrightarrow{x_1},
  \overrightarrow{x_2},
  \overrightarrow{x_3})] ~\lor \\
    (\forall \overrightarrow{x_1})
  (\exists \overrightarrow{x_2})
  (\forall \overrightarrow{x_3}) [ F(\overrightarrow{x_1},
  \overrightarrow{x_2},
  \overrightarrow{x_3})] ~\lor \\
    (\forall \overrightarrow{x_1})
  (\forall \overrightarrow{x_2})
  (\exists \overrightarrow{x_3}) [ F(\overrightarrow{x_1},
  \overrightarrow{x_2},
  \overrightarrow{x_3})] ~\lor \\
    (\forall \overrightarrow{x_1})
  (\forall \overrightarrow{x_2})
  (\forall \overrightarrow{x_3}) [ F(\overrightarrow{x_1}.
  \overrightarrow{x_2},
  \overrightarrow{x_3})].
\end{gather*}
(Since $\forall$ is never less demanding than $\exists$, the final
disjunct above doesn't affect the formula's truth value, and more
generally, we could just always consider just the quantifier
assignments that have exactly $j$
$\exists$ quantifiers, rather
than those that have $j$ or fewer $\exists$ quantifiers. But there is no
need to use that, so we will just ignore it.)

Notice that our giant disjunctive formula may
very well not be satisfied on the same set of inputs as the original
one, namely, it may be satisfied on additional inputs.
What we'll need in our proof is to
make an easily-computed cousin of the original formula that somehow is
``preinsulated'' from having its truth value changed by the
above type of quantifier (re)assignment, yet
that evaluates to true if and only if the original formula does.  
Let us give that cousin.  If the original formula is (as above)
  \[(Q_1 \overrightarrow{x_1})
  (Q_2 \overrightarrow{x_2})
  \cdots
  (Q_\ell \overrightarrow{x_\ell})  [ F(\overrightarrow{x_1},
  \overrightarrow{x_2},
  \dots,
  \overrightarrow{x_\ell})],\]
then our preinsulated cousin for it is
  \[(Q_1 \overrightarrow{x_1}, b_1)
  (Q_2 \overrightarrow{x_2}, b_2)
  \cdots
  (Q_\ell \overrightarrow{x_\ell}, b_\ell) \Big[
 \Big(F(\overrightarrow{x_1},
  \overrightarrow{x_2},
  \dots,
  \overrightarrow{x_\ell})\Big) ~\land   ~
  \Big(  \bigwedge_{\{ i \, \mid \: Q_i = \exists  \}} b_i\Big)
  \Big].\]
Note that this cousin's propositional statement 
$ \Big(F(\overrightarrow{x_1},
  \overrightarrow{x_2},
  \dots,
  \overrightarrow{x_\ell})\Big) ~\land   ~
  \Big(  \bigwedge_{\{ i \, \mid \: Q_i = \exists  \}} b_i\Big)$ 
  can never evaluate to true unless each $Q_i$ that is an $\exists$ in the
  original formula is an $\exists$ here.    Of course, that is true
  in this naked version of the cousin, since the quantifiers are the
  same and have not been shuffled yet.  However, the key point is
  that if we take this cousin, and as above take the disjunction of
  every possible assignment of its quantifiers that assigns at most
  $j$ of them to be $\exists$, then that new quantified boolean
  formula evaluates to true if and only if the original formula does.
  Further, the transformation from the original formula to the
  cousin (and here we really mean the cousin, not the larger
  disjunctive item) is clearly a polynomial-time transformation and
  does not change the quantifiers, and so in particular if the
  original quantifiers have each $Q_i$ being $\forall$ when $i$ is
  odd and $\exists$ when $i$ is even, then so will the cousin formula
  (and of course the $\ell$ and the $F$ of the cousin formula are the
  same as those of the original formula).

  Basically, we've made a propositional formula that nails down its existential
  quantifiers as to how the formula can possibly be made true, even if
  it has to weather a disjunction over all possible shufflings of its
  quantifiers (even if that shuffling is allowed to---%
pointlessly---decrease the number of
  $\exists$ quantifiers).

  Fix any $k \geq 0$.  
  With the above ``(preinsulated) cousin'' work
  in hand, we can now give our proof that 
  there exists an
  election system $\cale$, whose unweighted and weighted winner problems
  are in polynomial time, such that $\onlinesystembk{\cale}{k}$ is
  $\pilevel{2k+1}$-hard.  We will do so by giving such an election
  system $\cale$ (which depends on $k$)
  and
  giving a 
  polynomial-time many-one reduction to $\onlinesystembk{\cale}{k}$
  from the 
  $\pilevel{2k+1}$-complete problem we will denote $A_{2k+1}$,
  where $A_{2k+1}$ is 
all formulas of the form
  \[
    (\forall \overrightarrow{x_1})
  (\exists \overrightarrow{x_2})
    (\forall \overrightarrow{x_3})
  (\exists \overrightarrow{x_4})
  \cdots
  (\forall \overrightarrow{x_{2k+1}}
  )  [
  F(\overrightarrow{x_1},
  \overrightarrow{x_2},
  \dots, \overrightarrow{x_{2k+1}}
  )]\]
that evaluate to true (here, $F$ is required to be a propositional
boolean formula and as above, the $\overrightarrow{x_i}$'s are
pairwise disjoint variable collections) and such that
(this additional nonstandard requirement clearly can be made without
loss of generality, in the sense that even with this the set
we are defining is clearly $\pilevel{2k+1}$-complete)
at least one variable from within each of the $\overrightarrow{x_i}$'s
occurs in the formula 
$  F(\overrightarrow{x_1},
  \overrightarrow{x_2},
  \dots, \overrightarrow{x_{2k+1}}
  )$.

\subsubsection{The Election System $\cale$}

  Let us give the election system $\cale$ and the polynomial-time
  many-one reduction
  from $A_{2k+1}$ to $\onlinesystembk{\cale}{k}$.
  As mentioned before, we will stay as close as possible to the
  argument line used in the hardness arguments for online manipulation.

Let us first specify our (unweighted) election system $\cale$.  If the
input to $\cale$ is $(C,V)$, the election system will do the following. 
The system will first look at the candidate set and determine which
candidate, let us call it $c$, has the lexicographically smallest among
the candidate names in $C$.  Next, $\cale$ will look at the bit-string $c$
to determine whether it is (i.e., whether it encodes)
a \emph{tiered boolean formula}.  A tiered
boolean formula~\cite{hem-hem-rot:j:online-manipulation} is a formula
whose 
  variable names are each of the form $x_{i,m}$ (which really means a
  direct encoding of a string, such as ``$x_{14,92}$''); the $i,m$
  fields must be positive integers.  If $c$ cannot be validly
  parsed in this way, then we declare no candidates to be
  the winner of the election.  But otherwise, we have in hand
  a tiered formula represented by $c$.  For a given $i$, we will
  think of all the variables $x_{i,m}$ that occur in $c$ as being a ``block''
  of variables (as they will all be falling under a particular
  quantifier).
  Let $\blocks$ denote the maximum value that occurs in the first
  component of any of our $x_{i,m}$ variables that occur in $c$
  (technically, in the formula encoded by $c$, but we will henceforward
  just refer to that as $c$); this will tell us the number of blocks.
  Let $\width$ denote the maximum value that occurs in the second
  component of any of our $x_{i,m}$ variables that occur in $c$.
  Now, $\cale$ will immediately declare that everyone loses unless
  all 
  of the following things hold:
  \begin{enumerate}
  \item The number of voter names in $V$ is at least
    $\blocks + 1 $.
    (That is, each vote consists of a name and an order.  It is possible
    that the same name appears multiple times, perhaps some times with
    the same order---as happens when an unweighted winner problem
    is created from a weighted election that has at least one
    weight that is 2 or more---and perhaps sometimes with different orders.
    But what this condition is saying is that if one considers the set of
    all names that occur in at least one of the votes in $V$,
    that set of names is of cardinality at least $\blocks + 1$.)
\item $\blocks =  2k+1$.

\item The number of candidates in $C$ is
greater than or equal to 
$1 + 2 \cdot \width $.
(This condition is to ensure that
        each vote's preference order is about a large enough number of
        candidates that it can be used to assign all the variables
        in one quantifier block.)
        
        \item No block is unpopulated; that is, for each 
          $b$, $1 \leq b \leq \blocks$, $c$ 
          contains at least one variable $x_{i,m}$ whose
          first component is $b$.
        \end{enumerate}
        Now, $\cale$ will make a list, which we will refer to
        as the
        ``special list,'' of $\blocks$ votes from
        $V$.  $\cale$  will make the list in the following somewhat involved fashion.
        Let $\mathit{Names}$ be the set of all voter names that occur
        in $V$, i.e., it is the set of all names
that are
        the first component of at least one vote in $V$.
        Recall that if we have reached this point, we know that
        $\|\mathit{Names}\| \geq \blocks+1$.
Our list
         will not include any vote whose first component is
        the lexicographically smallest string in $\mathit{Names}$.
        (This feature will be used later in the proof to keep the
        pathological voter $u$ from breaking the proof.)
        For each $i$, $1 \leq i \leq \blocks$, the $i$th vote in our
        special list will be, among all votes whose first component
        is lexicographically the $(i+1)$st smallest string in $\mathit{Names}$,
        the vote with the lexicographically smallest second component.
        Note that, for example,
        even if ``Alice''
        appears many times as the first-component field
        of elements in the list $V$,
        at most one Alice vote will appear in our special list of
        $\blocks$ votes. 

        Now (recall, we are still defining the election system $\cale$'s action; we
        are not here speaking of our online bribery problem),
        $\cale$ will use the vote of the first voter
  in this special list to assign truth values to all
  variables of the form~$x_{1,\wasast}$, will use the vote of the 
  second voter in this special list to assign
  truth values to all variables of the form~$x_{2,\wasast}$, and so on up to the vote of the
  $\blocks$th voter, which will assign truth values to all
  variables of the form~$x_{\blocks,\wasast}$.

  Let us now set out how votes create those assignments.  For this, we
  will use the coding scheme
  from
  Hemaspaandra, Hemaspaandra, and Rothe~\cite{hem-hem-rot:j:online-manipulation}, which is as follows.
  Consider a vote whose total order over $C$ is $\sigma'$ (recall that
  we have that $\|C\| \geq 1 + 2 \cdot \width$).  Remove $c$ (the candidate
encoding the formula) from the
  order $\sigma'$, yielding~$\sigma''$.  Let
  $c_1 <_{\sigma''} c_2 <_{\sigma''} \cdots <_{\sigma''} c_{2 \cdot
    \width}$ be the $2 \cdot \width$ least preferred candidates in
  $\sigma''$.  Then $\cale$ will build a vector in $\{0,1\}^{\width}$
  as follows:
  The $d$th bit of the vector is $0$ if the string that
  names
  candidate $c_{2d - 1}$ (e.g., ``Eilonwy'')
  is lexicographically less than the string
  that names candidate $c_{2d}$
 (e.g., ``Rowella''),
  and this bit is $1$ otherwise.
  Let $\overrightarrow{b_i}$ denote the vector thus built from the $i$th vote (in the above
special list), $1 \leq i \leq \blocks$.  For each variable
$x_{i,m}$ occurring in~$c$, assign to that variable the value
of the $m$th bit
of~$\overrightarrow{b_i}$,
where $0$ represents \emph{false} and $1$ represents
  \emph{true}.  This process has assigned all the variables of~$c$, so $c$
  evaluates to either \emph{true} or \emph{false}.  If $c$ evaluates
  to \emph{true}, everyone wins under system $\cale$, and
  otherwise everyone loses under system $\cale$.  This completes
  our definition of election system~$\cale$.
  Note that the election system $\cale$ that we just defined
  has a polynomial-time winner problem in both the unweighted
  and the weighted cases.  This is because our special list
  was simply constructed and each boolean formula,
  given an assignment for all its variables,
  can easily be evaluated in polynomial time.

\subsubsection{The Reduction}

With our $\cale$ in hand, 
let us now show that $\onlinesystembk{\cale}{k}$ is
$\pilevel{2k+1}$-hard, via
  giving a 
  polynomial-time many-one reduction to $\onlinesystembk{\cale}{k}$
  from the 
  $\pilevel{2k+1}$-complete problem $A_{2k+1}$ that was defined above.

Here is how the reduction works.
Let $y$ be an instance of $A_{2k+1}$, i.e., $y$ is of the form:
  \[
    (\forall \overrightarrow{x_1})
  (\exists \overrightarrow{x_2})
    (\forall \overrightarrow{x_3})
  (\exists \overrightarrow{x_4})
  \cdots
  (\forall \overrightarrow{x_{2k+1}}
  )  [
  F(\overrightarrow{x_1},
  \overrightarrow{x_2},
  \dots, \overrightarrow{x_{2k+1}}
  )]\]
and at least one variable from within each of the $\overrightarrow{x_i}$'s
occurs in $F$.  The above formula may evaluate to true or
may evaluate to false; that is the question we want answered.  (We
here are ignoring syntactically illegal inputs, as they obviously
are not in $A_{2k+1}$ and so can be easily handled.)  Let us transform
$y$ into its cousin formula, as per the above discussion, i.e., we
transform $y$ into the following formula $y'$:
\begin{align*}
\MoveEqLeft[4]    (\forall \overrightarrow{x_1},b_1)
  (\exists \overrightarrow{x_2},b_2)
    (\forall \overrightarrow{x_3},b_3)
  (\exists \overrightarrow{x_4},b_4)
  \cdots
  (\forall \overrightarrow{x_{2k+1}},b_{2k+1}
  )
  \\
  &
    \Big[
 \Big( F(\overrightarrow{x_1},
  \overrightarrow{x_2},
  \dots, \overrightarrow{x_{2k+1}}
  )\Big)  ~\land~
    \Big(b_2 \land b_4 \wedge \dots \wedge b_{2k}\Big)\Big].
\end{align*}
(For the $k=0$ case, simply skip the second conjunct. The odd $b_i$
are superfluous, but do no harm.)
Note that we can view the above $y'$ as being the formula (here
$\Phi$ will be the entire above propositional part):
\begin{align*}
\MoveEqLeft[6] (\forall\, x_{1,1},
x_{1,2}, 
\ldots , x_{1,k_1})\,
(\exists\, x_{2,1},
x_{2,2}, 
\ldots , x_{2,k_2})\, \cdots
(\forall_{2k+1}\, x_{2k+1,1},
x_{2k+1,2}, 
\ldots , x_{2k+1,k_{2k+1}})\,
\\
& [\Phi(x_{1,1},
x_{1,2},
\ldots , x_{1,k_1},
x_{2,1},
x_{2,2},
\ldots , x_{2,k_2}, \ldots ,
x_{2k_1,1},
x_{2k+1,2}, 
\ldots , x_{2k+1,k_{2k+1}})],
\end{align*}
where the quantifier applied to the $x_{i,\wasast}$
variables is $\forall$ if $i$ is odd and is $\exists$ if $i$ is even,
the $x_{i,m}$
are boolean variables,
$\Phi$ is the abovementioned propositional boolean formula,
and note that for each~$i$, $1 \leq i \leq 2k+1$, $\Phi$
will contain at least one
variable of the form~$x_{i,\wasast}$.

Let us, with $y'$ in hand, and viewed in terms of the just-given variable
names and formula~$\Phi$, continue with specifying the reduction, keeping
in mind that this is a preinsulated formula.

Our many-one reduction will map to the instance
$(C,V,\sigma,d,B)$ of 
$\onlinesystembk{\cale}{k}$, specified by the following:
\begin{enumerate}
\item $C$ contains a candidate whose name, $c$, encodes~$\Phi$, and in
  addition $C$ contains $2 \cdot \max(k_1, \ldots , k_{2k+1})$ other
  candidates, all with names lexicographically greater than~$c$---for
  specificity, let us say their names are the $2 \cdot \max(k_1,
  \ldots , k_{2k+1})$ strings that immediately follow $c$ in
  lexicographic order.

\item $V$ contains $2k+2$ voters, 
  $0, 1, 2, \ldots , 2k+1$, who vote in
  that order, where $u = 0$ is the distinguished voter (so there are
  no voters in $V_{<u}$ and there are $2k+1$ voters in $V_{>u}$).
  The voter names will be lexicographically ordered by their
  number, so $0$ is least and $2k+1$ is greatest.
  $u$'s preference order $\sigma$ will actually be irrelevant, because
  since $u$ will have the lexicographically smallest name among all
  voters, $u$'s vote will be ignored by our election system $\cale$;
  but for specificity, let us say $u$'s preference
  order is to simply rank the candidates in lexicographic order.

\item The briber's preference order $\sigma$ is to like
  candidates in the opposite of their lexicographic order.  In
  particular, $c$ is the briber's most preferred candidate.
  And we set $d$ to be $c$ (so the goal of the online bribery
  problem becomes simply to make $c$ be a winner).
\item $B$, the limit on the number of allowed bribes, is $k$.
\end{enumerate}
Note that this 
is a polynomial-time reduction.  And it 
follows from this reduction's construction and 
the definition of $\cale$
 that $y'$ is in $A_{2k+1}$ if
 and only if the thus-constructed $(C, V \sigma, d, B)$ is in
 $\onlinesystembk{\cale}{k}$.

 Why?  The online bribery problem will be asking whether there is some
 way of bribing at most $k$ of the voters so as to make $c$ be
 a winner.  Since this is an online bribery problem, we thus are
 existentially quantifying for each of those at most $k$ bribed voters as to his or her
 cast preference order, and for all the other voters we are universally
 quantifying as to their preference orders.  The preference order
 of the voter named $i$ (who will on this input
 be the $(i+1)$st voter on our ``special
 list'' defined earlier)
 will be controlling the setting of
 the variables in the $i$th block of $\Phi$.

 Now, all possible choices of where to bribe are allowed, as long as
 the total number of bribes is at most $B$ (i.e., at most $k$).
 However, the election system is routing these assigned variables to
 our preinsulated formula $\Phi$, and so by the properties of our
 preinsulated formula, the only possible case that can result in the
 formula being true (and thus all candidates---and so in particular
 candidate $c$---winning, as opposed to all candidates losing) is when
 every even-numbered voter other than voter $0$
 is bribed.  But as there are only $k$
 allowed bribes and there are $k$ even-numbered voters other than
 $0$ (namely, the voters named $2,4,\dots,2k$), those
 bribes already use every
 allowed bribe, and so we know that the only way $c$ can be a winner
 is if every odd-numbered voter and voter $0$ remain unbribed and every
 even-numbered voter other than 0 is bribed.  And so our online-bribery problem is
 in fact testing whether $y \in A_{2k+1}$; $d$ is a winner (in fact,
 all candidates are winners) in the constructed instance
 $(C, V \sigma, d, B)$ if and only if $y \in A_{2k+1}$.

 We thus have shown that  $\onlinesystembk{\cale}{k}$ is
 $\pilevel{2k+1}$-hard, and thus in light of the earlier upper bound,
 have shown that $\onlinesystembk{\cale}{k}$ is
 $\pilevel{2k+1}$-complete.

\subsubsection{The Other Fifteen Cases}
\label{s:fifteen}

 It follows immediately from   Proposition~\ref{p:reductions}
 that $\onlinesystempbk{\cale}{k}$ is
 $\pilevel{2k+1}$-hard, and thus in light of the earlier upper bound
 is  $\pilevel{2k+1}$-complete.

 As to $\onlinesystemwbk{\cale}{k}$, if we use the entire above
 construction and assign to each voter weight 1, we have (keeping in
 mind that for $\cale$ both the unweighted and the weighted winner
 problems are in polynomial time), the thus-altered construction
 shows that 
$\onlinesystemwbk{\cale}{k}$ is
 $\pilevel{2k+1}$-hard, and thus in light of the earlier upper bound
 is
 $\pilevel{2k+1}$-complete.  And from that, the earlier upper bound,
 and Proposition~\ref{p:reductions}, we have that 
 $\onlinesystempwbk{\cale}{k}$ is
 $\pilevel{2k+1}$-complete.

 That covers the four constructive (i.e., not destructive) cases
 of part~\ref{p:lower-2} of the theorem.  

 But the same construction, easily modified for the case of
 unbounded numbers of quantifiers, analogously yields
 PSPACE-hardness results, and thus by the earlier upper bounds
 PSPACE-completeness results, for the four constructive cases
 from part~\ref{p:lower-unbounded} of the theorem.  In particular,
 to show that 
$\onlinesystemb{\cale}$ is
 $\pspace$-hard, and thus in light of the earlier upper bound PSPACE-complete,
 we now map to 
$\onlinesystemb{\cale}$
from the PSPACE-complete set $A_{\qbf}$, defined here as
 all formulas of the form
  \[
    (\forall_1 \overrightarrow{x_1})
  (\exists_2 \overrightarrow{x_2})
    (\forall_3 \overrightarrow{x_3})
  (\exists \overrightarrow{x_4})
  \cdots
  (\forall_z \overrightarrow{x_{z}}
  )  [
  F(\overrightarrow{x_1},
  \overrightarrow{x_2},
  \dots, \overrightarrow{x_{z}}
  )]\]
that evaluate to true (here, $F$ is required to be a propositional
boolean formula, and the $\overrightarrow{x_i}$'s are
required to be pairwise disjoint variable collections, and $z$ is an odd
integer) and such that
at least one variable from within each of the $\overrightarrow{x_i}$'s
occurs in the formula 
$  F(\overrightarrow{x_1},
  \overrightarrow{x_2},
  \dots, \overrightarrow{x_{z}}
  )$.
  Specifying that the leading quantifier is a $\forall$ and that
  the number of quantifiers in our alternating quantifier sequence
  is odd and that at least one variable from each block 
  occurs in $F$ is not the standard version of QBF, but clearly also yields a
  PSPACE-complete set.  And having it be of this form makes it clear
  how to specify $\cale$ (namely, for inputs
  whose $\blocks$ value
  is $2k+1$, we use the version of the above $\cale$ that assumes and enforces
  that $\blocks = 2k+1$) and what reduction to use (namely, given
  a formula with the syntax and properties (except perhaps truth)
  of $A_{\qbf}$, having $2k+1$ alternating
  quantifiers, we use the actions of the reduction for the $\pilevel{2k+1}$ case
  above).  (The reduction's actions are sufficiently uniform and
  simple that what was just mentioned can itself be done in a single
  polynomial-time many-one reduction that handles all odd sequence
  lengths of alternating quantifiers whose first quantifier
  is a $\forall$.)

  So we have handled all eight constructive cases.\footnote{For (no pun intended) completeness, we mention
    that we could have alternatively established the lower bounds for
    $\onlinesystempbk{\cale}{k}$ and
    $\onlinesystempb{\cale}$ 
    via results proved or stated in 
    Hemaspaandra, Hemaspaandra, and Rothe's~\cite{hem-hem-rot:j:online-manipulation}
    work on online manipulation, in light of the fact that one can
    simulate online manipulation by priced online
    bribery (via setting the budget to 1, 
    the price of the manipulators each to~0, and the price
    of the nonmanipulators each to~2).  That paper handles weights differently
    than this paper, and it doesn't provide lower-bound matching 
    results for any of the general-case 
    destructive settings.  However, from what that paper does do one can,
    using the gateway we just mentioned, claim (from that paper's
    stated-without-proof
    result regarding the ``freeform online manipulation
    problem''~\cite[p.~702]{hem-hem-rot:j:online-manipulation})
    the $\pilevel{2k+1}$-hardness of 
    $\onlinesystempbk{\cale}{k}$ for some $\cale$ whose unweighted
    winner problem is in polynomial time, and also one can claim
    (regarding the other part of our theorem) the 
    $\pspace$-hardness of 
    $\onlinesystempb{\cale}$ for some $\cale$ whose unweighted
    winner problem is in polynomial time.}

  The eight
  destructive cases are analogous.  We won't do this in detail,
  but basically as to $\cale$ one does everything as above, 
  except every place that in the definition of an election system
  $\cale$ we said everyone wins one changes that to saying that everyone
  loses, and everywhere we above in the definition of an election
  system $\cale$ said everyone loses one changes that to saying that
  everyone wins.  And as to the reductions, one uses the same
  reductions as above.

  This completes the proof of the theorem.
~\end{proof}

\section{Online Bribery for Specific Election Systems}\label{sec:online-brib-spec}

In this section, we look at the complexity of online bribery
for various natural systems.
For both Plurality and
Approval, we show
that priced, weighted online bribery is
$\np$-complete
but that the election system's other three online bribery variants are in~$\p$.
Since these other three problem variants of nonsequential bribery are
known to be $\np$-complete~\cite{fal-hem-hem:j:bribery},
this also shows that nonsequential bribery can be harder than
online bribery for natural systems.
In addition, we provide complete dichotomy theorems that distinguish 
NP-hard from easy cases for
all our online bribery problems for scoring protocols
and additionally we show that Veto
elections, even with three candidates, have even higher lower
bounds for weighted online bribery, namely $\p^{\np[1]}$-hardness.

The following theorem is useful for proving lower bounds for online
bribery for specific systems.\footnote{For an election
  system~$\mathcal{E}$, $\systemucm{\cal E}$ denotes the
  unweighted coalitional manipulation problem: Given a set $C$
  of candidates, a collection $V$ of nonmanipulative
  voters over $C$, a collection $W$ of manipulative voters
  (who will come in without specified preferences),
  and a
  designated candidate $c \in C$, can we assign
  preference orders over~$C$ to the members of $W$ in such a way that 
  $c$ is a winner of the election $(C,V \cup W)$?  If the voters
  are weighted, we obtain the
  weighted coalitional manipulation problem $\systemwcm{\cal E}$;
  note that the manipulators' weights but not their preferences
  are given in the problem
  instance.  The {destructive} variants of
  these two problems (where the goal is to prevent $c$ from
  being a winner in 
  the manipulated election) are denoted by $\systemducm{\cal E}$ and
  $\systemdwcm{\cal E}$, respectively.
\label{foo:ucm-wcm-ducm-dwcm}}

\begin{theorem}\label{t:m-to-b}
\begin{enumerate}
\item
Nonsequential manipulation reduces to corresponding online bribery.
(So, $\systemucm{\cal E}$ reduces to $\onlinesystemb{\cal E}$,
$\systemducm{\cal E}$ reduces to $\onlinesystemdb{\cal E}$,
$\systemwcm{\cal E}$ reduces to $\onlinesystemwb{\cal E}$, and
$\systemdwcm{\cal E}$ reduces to $\onlinesystemdwb{\cal E}$.)

\item
  Constructive manipulation in the unique winner model\footnote{In the
    \emph{unique winner model}, the goal of a constructive manipulation
    action is to have the designated candidate be the only winner of
    the manipulated election.  In the destructive case, the goal is to
    ensure that the designated candidate is not
    a 
unique winner of the
    manipulated election.}
reduces to corresponding online destructive bribery
(so, $\systemucm{\cal E}$ in the unique winner model
 reduces to $\onlinesystemdb{\cal E}$ and
$\systemwcm{\cal E}$ in the unique winner model
reduces to $\onlinesystemdwb{\cal E}$)
for election systems that always have winners (if there are candidates).

\item
Online manipulation reduces to corresponding online priced bribery.
(So, $\onlinesystemucm{\cal E}$ reduces to $\onlinesystempb{\cal E}$,
$\onlinesystemducm{\cal E}$ reduces to $\onlinesystemdpb{\cal E}$,
$\onlinesystemwcm{\cal E}$ reduces to $\onlinesystempwb{\cal E}$, and
$\onlinesystemdwcm{\cal E}$ reduces to $\onlinesystemdpwb{\cal E}$.)
\end{enumerate}
\end{theorem}

\begin{proof}
For the first part, 
let $V_{<u}$ be the nonmanipulators and
let $\{u\} \cup V_{>u}$ be the manipulators.
The vote of $u$ is irrelevant.
Let $k$ be the number of manipulators 
(meaning all voters in $\{u\} \cup V_{>u}$ can be bribed).
In the constructive case, the preferred candidate $p$ becomes the
designated candidate, which is ranked first in $\sigma$.
In the destructive case, the despised candidate $d$ becomes the
designated candidate, which is ranked last in $\sigma$.

For the second part, 
let again be $V_{<u}$ the nonmanipulators and
$\{u\} \cup V_{>u}$ the manipulators.
The vote of $u$ is irrelevant.
Let $k$ be the number of manipulators 
(meaning again that all voters in $\{u\} \cup V_{>u}$ can be bribed).
The ranking $\sigma$ puts the preferred candidate $p$ first. The other
candidates are ranked in lexicographic order, and the designated
candidate is the lexicographically smallest candidate in $C - \{p\}$, i.e.,
the candidate ranked second in $\sigma$.

For the last part, 
set the price of the manipulators to 0, the price of the nonmanipulators
to 1, and set $k$ to~0.%
{}~\end{proof}
It is interesting to note that, assuming $\p \neq \np$, 
bribery does not reduce to corresponding online bribery,
not even for natural systems. For example, 
Approval-Bribery is NP-complete~\cite[Theorem 4.2]{fal-hem-hem:j:bribery},
but we will show below in Theorem~\ref{t:approval} that
$\onlinesystemb{Approval}$ (and
even $\onlinesystemwb{Approval}$ and $\onlinesystempb{Approval}$)
are in $\p$.

We end this section with a simple observation about
unpriced, unweighted online bribery.
\begin{observation}\label{o:bribers-last}
For unpriced, unweighted online bribery, it is always optimal (meaning that if
the briber can reach his or her goal, it can be reached in this way) to bribe
the last $k$ voters (we don't even have to handle $u$ in a special
way). This implies that unpriced, unweighted online bribery
is certainly reducible to unweighted online
manipulation, and so we inherit those upper bounds.
\end{observation}

\subsection{Plurality}

In this section, we completely classify the complexity of all
our versions of online bribery for the most important natural 
system, Plurality. In this system, each candidate scores a
point when it is ranked first in a vote and the candidates
with the most points are the winners.  We show that these
problems are NP-complete if we have both prices and weights,
and in P in all other cases. 

The following observation is crucial in our upper
bound proofs:
For $\onlinesystempwb{Plurality}$ and
$\onlinesystemdpwb{Plurality}$,
there is a successful bribery if
and only if there is a successful bribery where all
bribed voters
from $u$ onward vote for the same
highest-scoring desired candidate\footnote{In the constructive
  case we call all members of $\{c \ | \ c \geq_\sigma d\}$---and
  in the destructive case  we call
all members of $\{c \ | \ c >_\sigma d\}$---\emph{desired} candidates, 
where $\sigma$ 
  is the briber's ideal ranking and $d$ the designated candidate.}
and all nonbribed voters
after $u$ vote for
the same so-far highest-scoring undesired candidate.
If $u$ is bribed, we
do not count $u$'s original vote to compute the highest score. If
$u$ is not bribed, then we count $u$'s vote.

\begin{theorem}\label{t:plurality-easy}
$\onlinesystemb{Plurality}$,
$\onlinesystemdb{Plurality}$,
$\onlinesystemwb{Plurality}$,
$\onlinesystemdwb{Plurality}$,
$\onlinesystempb{Plurality}$, and
$\onlinesystemdpb{Plurality}$
are in \p.
\end{theorem}

\begin{proof}
First look at $\onlinesystemwb{Plurality}$.
We are given an OBS $(C, V, \sigma, d, k)$, where $V = (V_{<u}, u, V_{>u})$.
We can bribe successfully if and only if we can bribe successfully and we 
bribe $u$ or we can bribe successfully and we do not bribe $u$.
Let $\Gamma_d = \{c \ | \ c \geq_\sigma d\}$ be the desired candidates and
let $k'$ be the number of voters in $V_{<u}$ that have already been bribed.
If $k' > k$, our instance is illegal and we reject. 
If $\Gamma_d = C$, we can successfully bribe, and we accept.

To check whether we can bribe successfully and bribe $u$, let
$c$ be a candidate in $\Gamma_d$ with highest score in $V_{<u}$.
If $k' \geq k$, 
we can't bribe $u$ and we reject. Otherwise, 
bribe $u$ to vote for $c$ and bribe 
$\min(k - k' - 1, \|V_{>u}\|)$ highest-weight voters in  $V_{>u}$
to vote for $c$. This will give us the score of $c$ after bribery.
Now let $h$ be a candidate in $C - \Gamma_d$
with highest score in $V_{<u}$. 
Assume that all the nonbribed voters in $V_{>u}$ vote
for $h$. Then we can successfully bribe if and only if the score of $c$ is
at least the score of $h$.

To check whether we can bribe successfully and not bribe $u$, let
$c$ be a candidate in $\Gamma_d$ with highest score in $V_{<u} \cup \{u\}$.
Bribe $\min(k - k', \|V_{>u}\|)$ highest-weight voters in  $V_{>u}$
to vote for $c$. This will give us the score of $c$ after bribery.
Now let $h$ be a candidate in $C - \Gamma_d$
with highest score in $V_{<u} \cup \{u\}$.
Assume that all the nonbribed voters in $V_{>u}$ vote
for $h$. Then we can successfully bribe if and only if the score of $c$ is
at least the score of $h$.

For the destructive case, we argue similarly, except that we
let
$\Gamma_d = \{c \ | \ c >_\sigma d\}$
and we are successful if the score of
$c$ is greater than the score of $h$.
Next look at $\onlinesystempb{Plurality}$.
We argue similarly as in the weighted case, except that
we are now looking at bribery budgets instead of bribery limits, and 
so $k'$ is the price of the voters in $V_{<u}$ that are bribed, we make
sure that the price of the bribed voters is not higher than $k'$, and we
bribe the lowest-priced voters.

For completeness, we give the complete proof.
We are given an OBS $(C, V, \sigma, d, k)$, where $V = (V_{<u}, u, V_{>u})$.
We can bribe successfully if and only if we can bribe successfully and we 
bribe $u$ or we can bribe successfully and we do not bribe $u$.
Let $\Gamma_d = \{c \ | \ c \geq_\sigma d\}$ and
let $k'$ be the price of the voters in $V_{<u}$ that are bribed.
If $k' > k$, our instance is illegal and we reject. 
If $\Gamma_d = C$, we can successfully bribe, and we accept.

To check whether we can bribe successfully and bribe $u$, let
$c$ be a candidate in $\Gamma_d$ with highest score in $V_{<u}$.
If $k' + \pi(u) > k$, 
we can't bribe $u$ and we reject. Otherwise, 
bribe $u$ to vote for $c$ and bribe
a lowest-priced voter in $V_{>u}$ to vote for $c$,
as long as the price of the bribed voters is at most $k$.
This will give us the score of $c$ after bribery.
Now let $h$ be a candidate in $C - \Gamma_d$
with highest score in $V_{<u}$. 
Assume that all the nonbribed voters in $V_{>u}$ vote
for $h$. Then we can successfully bribe if and only if the score of $c$ is
at least the score of $h$.

To check whether we can bribe successfully and not bribe $u$, let
$c$ be a candidate in $\Gamma_d$ with highest score in $V_{<u} \cup \{u\}$.
Bribe a lowest-priced voter in $V_{>u}$ to vote for $c$,
as long as the price of the bribed voters is at most $k$.
This will give us the score of $c$ after bribery.
Now let $h$ be a candidate in $C - \Gamma_d$
with highest score in $V_{<u} \cup \{u\}$.
Assume that all the nonbribed voters in $V_{>u}$ vote
for $h$. Then we can successfully bribe if and only if the score of $c$ is
at least the score of $h$.

For the destructive case, we again argue similarly, except that we
let
$\Gamma_d = \{c \ | \ c >_\sigma d\}$
and we are successful if the score of
$c$ is greater than the score of $h$.%
{}~\end{proof}

\begin{theorem}\label{t:plurality-hard}
$\onlinesystempwb{Plurality}$ and
$\onlinesystemdpwb{Plurality}$ are \np-complete,
even when restricted to two candidates.
\end{theorem}

\begin{proof}
These problems are in NP, using the observation at the start
of this section: Guess a set of voters to bribe, 
check that their price is within the budget,
let all these bribed voters vote for the same highest-scoring desired candidate,
and let all nonbribed voters vote for the same highest-scoring
undesired candidate.
Accept if a desired candidate wins in the constructive case
and accept if no undesired candidate wins in the destructive case.

To show NP-hardness for the constructive case,
we use the same construction as for nonsequential bribery
with prices and weights. We reduce
from (the standard NP-complete problem)
Partition. Let $s_1, \ldots, s_n$ be a sequence of nonnegative
integers such that $\sum_{i=1}^n = 2S$. 
We map to OBS $(C, V, \sigma, d, k)$, 
where $C = \{d,c\}$, $d >_\sigma c$, the price and weight of the $i$th voter
are both $s_i$, $u$ is the first voter and votes for $c$,
and $k = S$.

For the destructive case, our designated candidate will be $c$, and
$V_{<u}$ consists of one unbribed weight-1 voter who votes for $d$.~\end{proof}

\subsection{Beyond Plurality}

A scoring protocol is
a vector  $\alpha = (\alpha_1, \ldots, \alpha_m)$ of integers
$\alpha_1 \geq \alpha_2 \geq \cdots\geq\alpha_m \geq 0$.
This defines an election system on $m$ candidates
where each candidate earns $\alpha_i$ points for each vote
that ranks it in the $i$th position and
the candidates with the most points are the winners. 

\begin{theorem}\label{t:scoring-dichotomy}
For each scoring vector $\alpha = (\alpha_1, \ldots, \alpha_m)$,
\begin{enumerate}
\item
$\onlinesystempwb{\alpha}$
and $\onlinesystemdpwb{\alpha}$ are in \p\ if $\alpha_1 = \alpha_m$ and
$\np$-hard otherwise;
\item
$\onlinesystemwb{\alpha}$ and
$\onlinesystemdwb{\alpha}$ 
are in \p\ if $\alpha_2 = \alpha_m$ and 
$\np$-hard otherwise; and 
\item
$\onlinesystemb{\alpha}, \onlinesystemdb{\alpha},
\onlinesystempb{\alpha}$, and $\onlinesystemdpb{\alpha}$
are in \p.
\end{enumerate}
\end{theorem}

Note that Theorem~\ref{t:scoring-dichotomy}
implies the lower bound of Theorem~\ref{t:plurality-hard}, but
does not imply Theorem~\ref{t:plurality-easy},
since in that theorem the number of candidates is not fixed.
Theorem~\ref{t:scoring-dichotomy} does apply to 3-candidate Veto, 
for which we will additionally prove higher lower bounds in Theorem~\ref{t:veto}.

\begin{proof}
If $\alpha_1 = \alpha_m$, all candidates are always winners.
If $\alpha_1 > \alpha_2 = \alpha_m$, this is in essence Plurality,
which is handled in the theorems above.
In all other constructive cases, the
hardness follows from the hardness for
$\systemwcm{\alpha}$
from
Hemaspaandra and Hemaspaandra~\cite{hem-hem:j:dichotomy-scoring} and Theorem~\ref{t:m-to-b}.
In all other destructive cases, the
hardness follows from the hardness for
$\systemwcm{\alpha}$ in the unique winner model
from
Hemaspaandra and Hemaspaandra~\cite{hem-hem:j:dichotomy-scoring} and Theorem~\ref{t:m-to-b}.
(Note that
$\systemdwcm{\alpha}$ is easily seen to be in \p.  Basically, to
make $d$ not a winner, we need an $a$ such that $a$'s score
is higher than $d$'s score. So, make all manipulators
vote $a > \cdots > d$, and compute the scores.
So, this does not help us with the lower bound for $\onlinesystemdpwb{\alpha}$.)

For the last case, we first look at the unpriced, unweighted
case, even though the result follows from the priced case.
We do this because the algorithm for this case is much simpler.
Note that it follows from 
Observation~\ref{o:bribers-last} that we can assume that all bribed voters go
last and so we have $k'$ nonbribed voters followed by $k$ bribed voters. Since
there are only a constant number of different votes (since $m$ is a constant), simply brute-force
to determine whether it is the case that for all (polynomially many)
possible $k'$ votes, there are $k$ votes such that a
candidate that is preferred to $d$ is a winner (for constructive)
or such that no candidate that $d$ is preferred to
is a winner.

Note that if we have prices, we can not assume that the bribed voters
come last. We also can not assume that we bribe the cheapest voters,
since later voters have more power than earlier voters, and so
a more expensive later voter could be a better choice to bribe than a
cheaper earlier voter. Still, we can solve the priced cases in 
polynomial time, by using dynamic programming. It is crucial that
the scores for unweighted elections are $O(\log n)$.

We want to compute $\pi(s_1, \ldots, s_m, u, k)$ to be the minimum
budget such that if the score of $c_i$ before $u$ is $s_i$ and 
there are $k$ voters after $u$, the briber can accomplish their goal.
Note that there are a constant number of votes for $u$ and that all other
numbers are $O(\log n)$. Also note
that $\pi(s_1, \ldots, s_m, u, k)$ is the minimum of 
the minimum budget needed when $u$ is bribed and
the minimum budget needed when $u$ is not bribed.
To compute the minimum budget needed when $u$ is bribed,
compute the minimum over all votes $v$ and $v'$
of $\pi(u) + \pi(s'_1, \ldots, s_m', v', k - 1)$, where
$s_i'$ is the score of $c_i$ from all but the last $k$ voters when
$u$ is bribed to vote $v$.
To compute the minimum budget needed when $u$ is not bribed,
compute the minimum over all votes $v'$
of $\pi(s'_1, \ldots, s_m', v', k - 1)$, where
$s_i'$ is the score of $c_i$ from all but the last $k$ voters when
$u$ is not bribed.~\end{proof}

In Veto, each candidate scores a point when it is not
ranked last in a vote and the candidates
with the most points are the winners.
Now let's look at 3-candidate-Veto.

\begin{theorem}\label{t:veto}
\begin{enumerate}
\item
$\onlinesystemb{\mbox{\rm 3-candidate-Veto}}$,
$\onlinesystemdb{\mbox{\rm 3-candidate-Veto}}$,
$\onlinesystempb{\mbox{\rm 3-candidate-Veto}}$, and
$\onlinesystemdpb{\mbox{\rm 3-candidate-Veto}}$ are in \p.
\item
$\onlinesystemwb{\mbox{\rm 3-candidate-Veto}}$ and
$\onlinesystemdwb{\mbox{\rm 3-candidate-Veto}}$ are
$\p^{\np[1]}$-complete.
\item
$\onlinesystempwb{\mbox{\rm 3-candidate-Veto}}$
and
$\onlinesystemdpwb{\mbox{\rm 3-candidate-Veto}}$
are $\p^{\np[1]}$-hard and
in $\deltatwo$ (and we conjecture that they are
$\deltatwo$-complete).
\end{enumerate}
\end{theorem}

\begin{proof}
The first part follows immediately from Theorem~\ref{t:scoring-dichotomy}.

For the second part, we 
look at different cases
for the placement of the designated candidate in the preference order
$a >_\sigma b >_\sigma c$.

\begin{itemize}
\item
If $d = c$ in the constructive case
or $d = a$ in the destructive case, the problem is trivial.
\item
If $d = a$ in the constructive case, we 
need to ensure that $a$ is a winner, and
if $d = b$ in the destructive case, 
we need to ensure that $a$ is the unique winner.
In both these cases, the nonbribed voters veto $a$, no matter what.
This means that the location of the bribed voters doesn't matter (though
their weights will).
This is $\np$-complete. For the upper bound, 
guess a set of voters to bribe, check that they are within the bribery limit, 
guess votes for the bribed voters, and
have all nonbribed voters veto $a$. NP-hardness
follows from the NP-hardness for weighted manipulation~\cite{hem-hem:j:dichotomy-scoring}
plus the proof of Theorem~\ref{t:m-to-b}.

\item
If $d = b$ in the constructive case, the goal is to not have $c$
win uniquely, and if $d = c$ in the destructive case,
the goal is to have $c$ not win.
In this case, all bribed voters veto $c$ (no matter what).

This is coNP-complete. To show that the
complement is in NP, 
pick the $k$ heaviest voters to bribe. Then check if you can
partition the remaining voters to veto $a$ or $b$ in such a way
that $c$ wins uniquely in the constructive case or
that $c$ wins in the destructive case.
Note that it is always best for the
briber to bribe the $k$ heaviest voters: Swapping the weights of 
a lighter voter to be bribed with a heavier voter not to be bribed
will never make things worse for the briber.
To show hardness, note that the complement is basically
(the standard NP-complete problem) Partition.
\end{itemize}

Putting the three cases together, the unpriced, weighted case is $\p^{\np[1]}$-complete,
since it can be written as the union of a NP-complete set and
a coNP-complete set that are P-separable.\footnote{Two sets $A$ and $B$
  are P-separable if there exists a set $X$ computable in polynomial time
  such that $A \subseteq X \subseteq \overline{B}$
  (see~\cite{gro-sel:j:complexity-measures}).}

It remains to show the third part, i.e., the priced, weighted case. 
Note that this case inherits the $\p^{\np[1]}$-hardness
from the weighted case (in fact, it already inherited this from
online manipulation, using Theorem~\ref{t:m-to-b}). For the $\deltatwo$
upper bound, 
we again look at different cases
for the placement of the designated candidate in the preference order
$a >_\sigma b >_\sigma c$.

\begin{itemize}
\item
If $d = c$ in the constructive case
or $d = a$ in the destructive case, the problem is trivial.

\item
If $d = a$ in the constructive case, we 
need to ensure that $a$ is a winner, and
if $d = b$ in the destructive case, 
we need to ensure that $a$ is the unique winner.
In both these cases, the nonbribed voters veto $a$, no matter what.
This means that the location of the bribed voters doesn't matter (though
their prices and weights will).
All cases are in $\np$:
Guess a set of voters to bribe, check that they are within the budget, 
guess votes for the bribed voters, and
have all nonbribed voters veto $a$.

\item
If $d = b$ in the constructive case, the goal is to not have $c$
win uniquely, and if $d = c$ in the destructive case,
the goal is to have $c$ not win.
In this case, all bribed voters veto $c$ (no matter what).

To show the upper bound for the priced, weighted case, note that we need
to check that there exists a set of voters that can be bribed within
the budget such that if all bribed voters veto $c$, then for all
votes for the nonbribed voters, $c$ is not the unique winner (in the
constructive case) or not a winner (in the destructive case).
This is clearly in $\sigmatwo$. With some care, we can show that it is
in fact in $\deltatwo$. 
First use an NP oracle to determine the largest possible
total weight (within the budget) of bribed voters.
Then determine whether we should bribe $u$, by using the oracle
again to determine the largest possible total weight (within the budget)
of bribed voters, assuming
we bribe $u$. If that weight is the same as the previous weight,
bribe $u$. Otherwise, do not bribe $u$. Repeating this will give us a
set of voters to bribe of maximum weight. All these voters will 
veto $c$. It remains to check that for all votes for the nonbribed voters,
$c$ is not the unique winner (in the
constructive case) or not a winner (in the destructive case).
This takes one more query to an NP oracle.\qedhere
\end{itemize}
\end{proof}

We end this section by looking at 
approval voting.
In approval voting, each candidate scores a point when it is approved
in a vote and the candidates
with the most points are the winners.
Note that approval sets have arbitrary sizes.
Though Approval-Bribery is
NP-complete~\cite[Theorem 4.2]{fal-hem-hem:j:bribery},
we show that $\onlinesystemb{Approval}$ (and
even $\onlinesystemwb{Approval}$ and $\onlinesystempb{Approval}$)
are in $\p$. This implies that even for natural systems, 
bribery can be harder than online bribery (assuming $\p \neq \np$).

\begin{theorem}\label{t:approval}
\begin{enumerate}
\item
$\onlinesystemb{Approval}$,
$\onlinesystemdb{Approval}$,
$\onlinesystempb{Approval}$,
$\onlinesystemdpb{Approval}$,
$\onlinesystemwb{Approval}$, and 
$\onlinesystemdwb{Approval}$
are each in $\p$.
\item
$\onlinesystempwb{Approval}$ and
$\onlinesystemdpwb{Approval}$ are
each $\np$-complete.
\end{enumerate}
\end{theorem}

\begin{proof}
The upper bounds are immediate from the observation that it is optimal
for the briber to have the bribed voters
approve all desired candidates
and to have the nonbribed voters approve all other candidates. 
The NP-hardness for the priced, weighted cases follows
with the same reduction as for the priced, weighted cases
of online bribery for Plurality from Theorem~\ref{t:plurality-hard}.~\end{proof}

\section{Conclusions}
\label{sec:conclusions}

We have introduced a model of online, sequential bribery in voting and
have initiated the study of the complexity of the most natural
problems in this setting.  In particular, we have shown that even for
election systems whose winners can be determined in polynomial time,
in an online, sequential setting these bribery problems can be
complete for PSPACE or, when restricted to at most $k$ bribes,
for~$\pilevel{2k+1}$.

On the other hand, we have also shown that for some natural, important
election systems, namely Plurality, Approval, and $3$-candidate-Veto,
such a dramatic complexity jump does not occur, and we pinpoint the
complexity of their bribery problems.  Table~\ref{tab:comparison}
compares our complexity results on online bribery for these specific
natural election systems to the known complexity results for nonsequential
bribery.

\begin{table}[!t]
\small
  \centering
\begin{tabular}{l l l}  
{Problem} & {Online bribery complexity} & {Nonsequential bribery complexity}\\
\midrule
 \textit{Plurality}\\
\ \ $\systemb{}$ &P (Thm.~\ref{t:plurality-easy}) &P  \cite[Thm.~3.1]{fal-hem-hem:j:bribery} \\
\ \ $\systempb{}$ &P (Thm.~\ref{t:plurality-easy}) &P  \cite[Thm.~3.3]{fal-hem-hem:j:bribery} \\
\ \ $\systemwb{}$ &P (Thm.~\ref{t:plurality-easy}) &P  \cite[Thm.~3.3]{fal-hem-hem:j:bribery} \\
\ \ $\systempwb{}$ & NP-complete
(Thm.~\ref{t:plurality-hard}) &NP-complete  \cite[Thm.~3.2]{fal-hem-hem:j:bribery} \\
\textit{3-candidate-Veto}\\
\ \ $\systemb{}$ & P (Thm.~\ref{t:veto}) & P \cite[Thm.~4.13]{fal-hem-hem:j:bribery}\\
\ \ $\systempb{}$ & P (Thm.~\ref{t:veto}) & P \cite[Thm.~4.13]{fal-hem-hem:j:bribery}\\
\ \ $\systemwb{}$ & $\p^{\np[1]}$-complete (Thm.~\ref{t:veto}) &
NP-complete \cite[Thm.~4.9]{fal-hem-hem:j:bribery} \\
\ \ $\systempwb{}$ & $\p^{\np[1]}$-hard and
in $\deltatwo$ (Thm.~\ref{t:veto})& NP-complete \cite[Thm.~4.8]{fal-hem-hem:j:bribery} \\
\textit{Approval}\\
\ \ $\systemb{}$ &P (Thm.~\ref{t:approval}) &NP-complete  \cite[Thm.~4.2]{fal-hem-hem:j:bribery}\\
\ \ $\systempb{}$ &P (Thm.~\ref{t:approval}) &NP-complete  \cite[Thm.~4.2]{fal-hem-hem:j:bribery}\\
\ \ $\systemwb{}$ &P (Thm.~\ref{t:approval}) &NP-complete  \cite[Thm.~4.2]{fal-hem-hem:j:bribery}\\
\ \ $\systempwb{}$ & NP-complete (Thm.~\ref{t:approval}) & NP-complete 
\cite[Thm.~4.2]{fal-hem-hem:j:bribery}\\
\end{tabular}
  \caption{\label{tab:comparison}
  Comparing the complexity of online bribery and nonsequential bribery
  for natural election systems}
\end{table}

A very natural direction for further research is to investigate the
complexity of online bribery in natural election systems other than
those studied in Section~\ref{sec:online-brib-spec} of this paper.
A more narrow, focused challenge would be to close the gap (see
Theorem~\ref{t:veto}),
for online-3-candidate-Veto-Weighted-\$Bribery
and
online-3-candidate-Veto-Destructive-Weighted-\$Bribery,
 between the $\p^{\np[1]}$-hardness lower bounds and
 the 
 $\deltatwo$ upper bounds; we conjecture that both problems in fact are
$\deltatwo$-complete.
The argument made in Footnote~\ref{fn:binary}
notwithstanding, another potential direction for additional study would
be to investigate the complexity of
online bribery when prices and/or weights are encoded
in unary rather than binary.  
A broad area for further research would be to
investigate online bribery
in a
probabilistic setting, perhaps with
probabilistic distribution information about future votes
and their correlations; that would have to be done 
keeping in mind
that later votes could be influenced by earlier
cast votes.
Finally, it
could
be
interesting to extend our model of online, sequential
bribery to other bribery variants such as swap bribery or shift bribery.

\section*{Acknowledgments}
A preliminary version of this paper appeared in the
    Seventeenth Conference on Theoretical Aspects of Rationality
    and Knowledge
    (TARK~2019)~\protect\cite{hem-hem-rot:c:online-bribery}.
We are very grateful to the anonymous conference and journal reviewers
and Eric Pacuit
for many helpful comments and suggestions.

\appendix

\section{Appendix: A Discussion of the
  Weighted Version of a Given Election System}\label{sec:disc-weight-vers}
Let us discuss the issue of weighted versions of election systems,
since the issue is not as straightforward as it might at first seem.
This is
a somewhat
rarefied,
model-focused discussion, and the reader
can safely skip this section unless interested in the issue of
models---and can rest assured that the outcome of the section is that
throughout this paper we use the notion of weighted versions
that is a common and intuitive one.

For natural systems, the weighted versions typically are already
defined.  In particular, usually typical and natural is to treat each
weight-$w$ vote as $w$ unweighted copies of that same vote, and to do
whatever the unweighted system would do given that particular
collection of votes.  Let us call this notion of an election system's
weighted version \emph{multiplicity expansion}.  This is simply interpreting
weights as what called 
the
``succinct'' (e.g.,~\cite{fal-hem-hem:j:bribery,fal-hem-hem-rot:j:llull,hem-hem-rot:j:hybrid,hem-hom:j:dodgson-greedy}
and many more)---as
opposed to the standard, aka ``nonsuccinct''---version
of an election's votes, which is using binary numbers to
represent the number of copies.

The unweighted winner problem is
in a strictly formal sense not 
a special case of
the defined-by-multiplicity-expansion weighted winner problem,
because the types of the voters differ (the former problem has no weights
and the latter problem has weights).  However, the
unweighted winner problem
polynomial-time reduces (and indeed reduces even via far more
restrictive reductions) to
the defined-by-multiplicity-expansion weighted winner 
problem simply by the near-trivial action
of setting the weight of each voter to be one.
So the
unweighted winner problem 
is in effect
never of greater complexity than the
defined-by-multiplicity-expansion weighted winner problem.
In particular,
if for an election system $\cale$ the 
defined-by-multiplicity-expansion weighted winner problem
is in $\p$, then so is the unweighted winner problem for $\cale$.

Note, however, that since the weights in our problems 
are in binary (see Footnote~\ref{fn:binary}), the
multiplicity-expansion
approach means that we may be simulating the original system on
an exponentially long input.  And so in this notion of
weighted elections it does \emph{not} in general
hold that if an (unweighted) election system has a polynomial-time
winner problem, then its weighted version---defined by multiplicity
expansion in the way just mentioned---has a polynomial-time winner
problem.  Indeed, it is very easy to make artificial examples of
election systems with polynomial-time winner problems where the winner
problem of the weighted version (as defined via multiplicity
expansion) is complete for exponential time, i.e., is complete for
${\rm E} = \cup_{c>0} \dtime [2^{cn}]$.\footnote{\label{f:E}For completeness,
  let us give an (admittedly artificial) construction of such
  a case.  Let $L_{\rm E}$ be any problem
  that is complete for E with respect to linear-time many-one reductions;
  it is well-known that E has such complete sets, e.g., the natural
  universal set $\{ \langle M,x,0^k \rangle \condition M$ accepts input
  $x$ within $k$ steps$\}$ (see~\cite{har:b:feasible-provable}).
  Let $\cale$ be the election system such that if there is exactly
  one candidate---$c$---in the election, and all voters in the election cast
  the same vote, and the number of voters is at least
  $2^{|c|}$, then $c$ wins if and only if $c \in L_{\rm E}$.
  Here, we are viewing $c$ as the bit-string naming that candidate, and
  $|c|$ denotes the number of bits in that string.
  In all other cases, all candidates win.
  The unweighted winner problem for this election
  system $\cale$ is clearly in~$\p$.
  However, note that the
    defined-by-multiplicity-expansion weighted winner problem of $\cale$
  is linear-time many-one complete for~$\rm E$.  Why?  It clearly
  is in~$\rm E$, since (though this is overkill) one can expand all
  weighted votes, that expansion transforms problem instances of size $n$ to
  problem instances of size at most $2^n$, and then one can accept
  if and only if the expanded instance is a member of the (polynomial-time
  solvable) unweighted winner problem of $\cale$.  But 
  the
      defined-by-multiplicity-expansion weighted winner problem of $\cale$
  is also many-one linear-time hard for $\rm E$, since
  $L_{\rm E}$ reduces to it by the many-one linear-time reduction
  that maps from $x$ to an election with one candidate, whose name is $x$,
  and with one voter, whose weight is $c'|x|$, where $c'$ is the
  smallest natural number such that $L_{\rm E} \in
  \dtime [2^{c'n}]$ (such a $c'$ exists by the definition of $\rm E$
  and the fact that $L_{\rm E} \in {\rm E}$).  Thus, the
  defined-by-multiplicity-expansion weighted winner problem of $\cale$
  is many-one linear-time complete for~$\rm E$, despite the fact
that the unweighted winner problem for $\cale$ belongs to $\p$.}
Despite that, natural
systems often have nice properties, and in particular, beautifully,
for a wide variety of natural systems, their weighted versions
are
defined by multiplicity expansion yet those versions
still have polynomial-time
winner problems.
For example, the family of election systems known
as Copeland and Llull
elections~\cite{cop:unpub:copeland,col-mcl-lor:j:medieval-social-choice}
(see also, e.g.,~\cite{fal-hem-hem-rot:j:llull}) has this property.
Even more importantly, so-called scoring systems have this
property (basically because, in their case, instead of doing an
exponential number of additions one at a time, one can simply
multiply).
And, for example, Section~\ref{sec:online-brib-spec} of
this paper discusses a number of such real-world systems, and
when speaking of their weighted versions is indeed doing so via
multiplicity expansion, and in doing so, will in each case
need to note that the thus-defined weighted version's winner problem
is in polynomial time.

In light of the previous paragraph, one still is left with the
question: Should we (a)~simply \emph{define} the weighted version of every
election system to be what is defined by the multiplicity expansion
approach above?  Or (b)~for each election system, should
each author, using
his or her human taste, 
hand-tailor, 
possibly in a quirky nonuniform way,
what he or she feels is
a natural notion if its weighted version?

Both those approaches have advantages and downsides.  We will discuss
those, and then---having made clear the downsides we are embracing---will
in this paper adopt approach~(a), namely, multiplicity expansion.
One disadvantage of multiplicity expansion is that
as mentioned above 
it in some
cases takes unweighted election systems whose winner problems are in
polynomial time and boosts the winner-problem complexity of their
weighted versions as high as exponential time.  But that is more a
feature to be aware of---and to avoid being bitten by due to
forgetting that the feature might exist---than a disadvantage.
The only true disadvantages we know of regarding using multiplicity
expansion as providing a general notion of interpreting weights in elections
are (i)~the approach is conflating the issue of weights with the issue of
succinct representations, and in particular
(ii)~for a few natural election systems,
the approach arguably
gets things wrong (and we will soon come back and discuss
this ``gets things wrong'' in more detail).
The advantage of the multiplicity-expansion 
approach is that it is usually highly natural, it is what almost anyone
would think of when asked what weight should be interpreted as, and
it gives an across-the-board, uniform approach to weight, rather than
having the notion be a system-by-system, ad hoc, argued-and-debated construct.
Turning to alternative~(b), since human tastes differ, 
alternative~(b)'s worst downside is that in the extreme one might be 
just basically 
hand in the weights to an election system's computation function and
let it do whatever it feels like with them---making them just
information/bits it can use, perhaps in utterly unnatural
ways.  That the definition of weight may not match
what people broadly
feel that weight means, and also (aside from the
promise that the humans hand-defining the weighted versions of
each system are using their taste) this approach does nothing 
to ensure that
the weighted version of an election systems has any connection
whatsoever to the system's unweighted version.  In light of this,
throughout this paper we outright define, for any (unweighted)
election system $\cale$, the weighted version of $\cale$ (which via
overloading we will also sometimes describe as $\cale$---the inputs to the system
though will make clear which version we are speaking of, since in one
case there are no weights and in the other case there are) via
multiplicity expansion.

It might seem strange for any paper to take a different approach
to
what
``weighted'' means.  Indeed, the literature's papers are generally
taking the multiplicity-expansion approach to weight as so natural and
compelling that they simply are employing it, generally without
mentioning it or mentioning why they are using it.
However, to further discuss this,
let us briefly come back to the worst weakness of
the choice we made, namely, that there are cases where this approach
arguably ``gets things wrong.''  In particular, the fact that the approach
conflates succinctness with weight means that the indivisibility of a
vote is being shattered.  Yet this is
a severe problem for
election systems that are very focused on actions
affecting individual votes.
For example, a famous election system known as Dodgson
elections~\cite{dod:unpub:dodgson-voting-system} is based on the
number of sequential exchanges of adjacent preferences within voters
needed to make a given candidate become
``a Condorcet winner''~\cite{con:b:condorcet-paradox},
i.e., a candidate who beats each other candidate in pairwise
head-on-head elections.  The multiplicity-expansion approach would
allow, for a example, some unweighted copies of a given weighted vote to have
an exchange made in them while other unweighted copies of that same
weighted vote
didn't have the exchange made in them;
yet that seems not to respect the spirit of
Dodgson's system.  (If one were hand-defining a notion of
weight for Dodgson's system, it is simply not clear
whether for a weight $i$ voter, now viewed as
utterly indivisible, one would want to view
an adjacent-candidates exchange as one exchange 
or $w_i$ exchanges.
Each could be argued for, and so there are at least two quite
different, quite reasonable notions of hand-defined weightedness for
Dodgson election.)  In Young
elections~\cite{you:j:extending-condorcet}, which are defined based on
how few deletions of voters are needed to make a given candidate a
(``weak,'' but let us not here worry about that distinction) Condorcet
winner, a similar issue occurs---similar both in the difficulty and in
the fact that there are, in a hand-built notion of weight for the
problem, two quite different notions that reasonable people could
disagree on as to which is more appropriate.  To the best of our
knowledge, the issue of how to frame weighted Dodgson and weighted
Young has been touched on only once in the literature, namely, it is
discussed in the technical-report version of a paper of Fitzsimmons
and Hemaspaandra~\cite{fit-hem:t:succinct}, although
not in the later versions of that same paper.  That paper provides a
valuable discussion of many of the issues discussed in this section,
and we highly commend it to the reader.
Interestingly, that paper finds only one natural election system for
which the weighted (in the sense of multiplicity expansion) winner
problem can even conditionally be shown to
be of greater complexity than the unweighted winner
problem, namely, the paper shows that to be the case for
the system known as Kemeny
elections~\cite{kem:j:mathematics-without-numbers}---which
are a case where the natural interpretation of weights
is via multiplicity expansion
(see~\cite{fit-hem:t:succinct})---if
a certain
complexity-theoretic conjecture holds (namely, that
there are sets acceptable via sequential access to $\np$
that are not acceptable via parallel access to
$\np$).
(In contrast,
footnote~\ref{f:E} of the present paper constructs
a case
where the weighted version of a problem is,
unconditionally, 
of greater complexity
than its unweighted version.)
The system Single Transferable Vote 
(see, e.g.,~\cite{bar-orl:j:polsci:strategic-voting})
has similar
issues to those of Dodgson and Young elections,
as to how to define weight for it.

To be clear, we are not
suggesting that multiplicity expansion captures the right notion of
weightedness for such systems as Dodgson's, Young's,
or Single
Transferable Vote.  It does not,
and papers studying those in weighted contexts should not employ
multiplicity expansion as their notion of weight.  However, what we
are saying is that multiplicity expansion is the best uniform, general
approach to weight, and so we use it here---while carefully making
sure not to apply it to cases such as Dodgson elections where it is
not a good fit, and making sure to always be aware that multiplicity
expansion can distort the complexity of winner problems and that one
must for whatever
cases one covers make sure not to casually
assume otherwise.  We refer the reader to
the work of Fitzsimmons and
Hemaspaandra~\cite{fit-hem:t:succinct,fit-hem:j:succinct-elections}
for further discussion of, and results on,
weights/succinctness/multiply in elections.

As a final comment, we
stress
that our definition regarding
how the winner problem is defined for the weighted version of an
(unweighted) election system in no way is
something
that binds the definitions of manipulative actions within those
systems.  Those actions and their costs are defined not by
the winner-problem handling but rather by the manipulative actions
themselves, e.g., in bribery of priced, weighted elections,
the price of a weighted voter is the cost of bribing that
particular voter (i.e., the weighted vote bribes or fails to be bribed
as a single unit).

\bibliographystyle{alpha}
%
%
%
\newcommand{\etalchar}[1]{$^{#1}$}

\end{document}